\documentclass{article}

\usepackage[utf8]{inputenc}

\usepackage{amsmath}
\usepackage[noblocks]{authblk}
\usepackage{natbib}
\usepackage{hyperref} 
\usepackage{amssymb}
\usepackage{verbatim}
\usepackage{amsthm}

\usepackage{algorithm}
\usepackage[noend]{algorithmic}
\usepackage{accents}
\usepackage{subcaption}

\newtheorem*{theorem*}{Theorem}
\newtheorem{proposition}{Proposition}[section]
\newtheorem{corollary}{Corollary}[section]

\newtheorem{lemma}{Lemma}

\newcommand{\ubar}[1]{\underaccent{\bar}{#1}}

\DeclareMathOperator{\argmin}{argmin}
\DeclareMathOperator{\argmax}{argmax}
\DeclareMathOperator{\SEP}{SEP}
\DeclareMathOperator{\DISP}{DISP}
\DeclareMathOperator{\DBCVI}{DBCVI}

\DeclareMathOperator{\cut}{cut}

\DeclareMathOperator{\evaluateCut}{evaluateCut}

\usepackage[left=2.5cm,top=3cm,right=2.5cm,bottom=3cm,bindingoffset=0.5cm]{geometry}

\usepackage{graphicx} 
\usepackage{caption}
\usepackage{subcaption}
\usepackage{booktabs}

\usepackage{algorithm}
\usepackage[noend]{algorithmic} 

\newtheorem{definition}{Definition}[section]

\newtheorem{theorem}{Theorem}

\begin{document}

\title{Graph-based Clustering under Differential Privacy}

\author[1,2]{Rafael Pinot \footnote{rafael.pinot@cea.fr}}
\author[1,2]{Anne Morvan \footnote{anne.morvan@cea.fr. Partly supported by the \textit{Direction G\'en\'erale de l'Armement} (French Ministry of Defense).}}
\affil[1]{CEA, LIST, 91191 Gif-sur-Yvette, France}
\affil[2]{Universit\'e Paris-Dauphine, PSL Research University, CNRS, LAMSADE, 75016 Paris, France}
\author[2]{Florian Yger}
\author[1]{C\'edric Gouy-Pailler}
\author[2]{Jamal Atif}

\date{March 10, 2018}

\maketitle

\begin{abstract}
In this paper, we present the first differentially private clustering method for arbitrary-shaped node clusters in a graph. This algorithm takes as input only an approximate Minimum Spanning Tree (MST) $\mathcal{T}$ released under weight differential privacy constraints from the graph. Then, the underlying nonconvex clustering partition is successfully recovered from cutting optimal cuts on $\mathcal{T}$. As opposed to existing methods, our algorithm is theoretically well-motivated. Experiments support our theoretical findings.
\end{abstract}

\section{Introduction}
Weighted graph data is known to be a useful representation data type in many fields, such as bioinformatics or analysis of social, computer and information networks. More generally, a graph can always be built based on the data dissimilarity where points of the dataset are the vertices and weighted edges express ``distances" between those objects. For both cases, graph clustering is one of the key tools for understanding the underlying structure in the graph~\citep{SCHAEFFER200727}. These clusters can be seen as groups of nodes close in terms of some specific similarity.

Nevertheless, it is critical that the data representation used in machine learning applications protects the private characteristics contained into it. Let us consider an application where one wants to identify groups of similar web pages in the sense of traffic volume {\em i.e.} web pages with similar audience. In that case, the nodes stand for the websites. The link between two vertices represents the fact that some people consult them both. Edge weights are the number of common users and thus, carry sensitive information about individuals. During any graph data analysis, no private user surfing behavior should be breached {\em i.e.} browsing from one page to another should remain private. 
As a standard for data privacy preservation, differential privacy~\citep{Dwork_2006} has been designed: an algorithm is differentially private if, given two close databases, it produces statistically indistinguishable outputs. Since then, its definition has been extended to weighted graphs. Though, machine learning applications ensuring data privacy remain rare, in particular for clustering which encounters severe theoretical and practical limitations. Indeed, some clustering methods lack of theoretical support and most of them restrict the data distribution to convex-shaped clusters~\citep{Nissim_2007,Blum:2008,McSherry_2009,Dwork:2011} or unstructured data~\citep{Ho:2013,Chen:2015}.
Hence, the aim of this paper is to offer a theoretically motivated private graph clustering. Moreover, to the best of our knowledge, this is the first weight differentially-private clustering algorithm able to detect clusters with an arbitrary shape for weighted graph data. 

Our method belongs to the family of Minimum Spanning Tree (MST)-based approaches. An MST represents a useful summary of the graph, and appears to be a natural object to describe it at a lower cost. For clustering purposes, it has the appealing property to help retrieving non-convex shapes~\citep{Zahn1971,Asano1988,MSTBasedClusteringAlgosICTAI2006, MorvanCGA17}. Moreover, they appear to be well-suited for incorporating privacy constraints as will be formally proved in this work.

\textbf{Contributions:} Our contributions are threefold: 1)  we provide the first theoretical justifications of MST-based clustering algorithms. 2) We endow \textsc{DBMSTClu} algorithm~\citep{MorvanCGA17}, an MST-based clustering algorithm  from the literature, with theoretical guarantees. 3)  We introduce a differentially-private version of \textsc{DBMSTClu}  and give several results on its privacy/utility tradeoff.

\section{Preliminaries}

\subsection{Notations} \label{sec:notations}
Let $\mathcal{G}=(V,E,w)$ be a simple undirected weighted graph with a vertex set $V$, an edge set $E$, and a weight function $w:= E \rightarrow \mathbb{R}$. 
One will respectively call the edge set and the node set of a graph $\mathcal{G}$ using the applications $E(\mathcal{G})$ and $V(\mathcal{G})$.
Given a node set $S\subset V$, one denotes by $\mathcal{G}_{|S}$ the subgraph induced by $S$.
We call $G=(V,E)$ the topology of the graph, and $\mathcal{W}_{E}$ denotes the set of all possible weight functions mapping $E$ to weights in $\mathbb{R}$. For the remaining of this work, cursive letter are use to represent weighted graphs and straight letters refer to topological arguments.
Since graphs are simple, the path $\mathcal{P}_{u-v}$ between two vertices $u$ and $v$ is characterized as the ordered sequence of vertices $\{u, \ldots, v  \}$.
We also denote $V_{\mathcal{P}_{u-v}}$ the unordered set of such vertices.
Besides, edges $e_{ij}$ denote an edge between nodes $i$ and $j$. 
Finally, for all positive integer $K$, $[K] := \{1, \ldots, K \}$. 

\subsection{Differential privacy in graphs}
As opposed to node-differential privacy~\citep{Kasiviswanathan2013} and edge-differential privacy~\citep{Hay2009}, both based on the graph topology, the privacy framework considered here is weight-differential privacy where the graph topology $G= (V,E)$ is assumed to be public and the private information to protect is the weight function $w:= E \rightarrow \mathbb{R}$.
Under this model introduced by~\cite{Sealfon_2016}, two graphs are said to be neighbors if they have the same topology, and \textit{close} weight functions. this framework allows one to release an almost minimum spanning tree with weight-approximation error of $O\left( |V|\log|E| \right)$ for fixed privacy parameters. Differential privacy is ensured in that case by using the Laplace mechanism on every edges weight to release a spanning tree based on a perturbed version of the weight function. The privacy of the spanning tree construction is thus provided by post-processing (cf. Th.~\ref{thm:postprocessing}). However, under a similar privacy setting,~\cite{Pinot2017} recently manages to produce the topology of a tree under differential privacy without relying on the post-processing of a more general mechanism such as the ``Laplace mechanism".
Their algorithm, called PAMST, privately releases the topology of an almost minimum spanning tree thanks to an iterative use of the ``Exponential mechanism" instead. For fixed privacy parameters, the weight approximation error is $O\left( \frac{|V|^{2}}{|E|}\log|V| \right)$, which outperforms the former method from~\cite{Sealfon_2016} on arbitrary weighted graphs under weak assumptions on the graph sparseness.
Thus, we keep here privacy setting from~\cite{Pinot2017}.

\begin{definition}[\cite{Pinot2017}] For any edge set $E$, two weight functions $w,w'\in \mathcal{W}_{E} $ are neighboring, denoted $ w \sim w'$, if 
$ || w - w' ||_{\infty} := \max\limits_{e \in E}| w(e) - w'(e) | \leq \mu .$
\end{definition}
$\mu$ represents the sensitivity of the weight function and should be chosen according to the application and the range of this function. 
The neighborhood between such graphs is clarified in the following definition.

\begin{definition} Let $\mathcal{G}=(V,E,w)$ and $\mathcal{G'}=(V',E',w')$, two weighted graphs, $\mathcal{G}$ and $\mathcal{G'}$ are said to be neighbors if $V=V'$, $E=E'$ and  $ w \sim w'$.
\end{definition}

The so-called weight-differential privacy for graph algorithms is now formally defined.

\begin{definition}[\cite{Sealfon_2016}]
For any graph topology $G= (V,E)$, let $\mathcal{A}$ be a randomized algorithm that takes as input a weight function  $w\in \mathcal{W}_{E}$. $\mathcal{A}$ is called $(\epsilon,\delta)$-differentially private on $G= (V,E)$ if for all pairs of neighboring weight functions $ w,w'\in \mathcal{W}_{E} $, and for all set of possible outputs $S$, one has 
$$\mathbb{P}\left[\mathcal{A}(w) \in S\right] \leq e^{\epsilon}\mathbb{P}\left[\mathcal{A}(w') \in S\right] + \delta.$$
If $\mathcal{A}$ is $(\epsilon,\delta)$-differentially private on every graph topology in a class $\mathcal{C}$, it is said to be $(\epsilon,\delta)$-differentially private on $\mathcal{C}$.
\end{definition}

One of the first, and most used differentially private mechanisms is the Laplace mechanism. It is based on the process of releasing a numerical query perturbed by a noise drawn from a centered Laplace distribution scaled to the sensitivity of the query.
We present here its graph-based reformulation.

\begin{definition}[reformulation \cite{Dwork_2006}]
Given some graph topology $G=(V,E)$, for any $f_{G}:\mathcal{W}_{E} \rightarrow \mathbb{R}^{k}$, the sensitivity of the function is defined as
$\Delta f_{G} =\max\limits_{w \sim w' \in \mathcal{W}_{E}}||f_{G}(w) -f_{G}(w')||_{1}.$
\end{definition}

\begin{definition}[reformulation \cite{Dwork_2006}]
Given some graph topology $G=(V,E)$, any function $f_{G}: \mathcal{W}_{E} \rightarrow \mathbb{R}^{k}$, any $ \epsilon > 0$, and $w \in \mathcal{W}_{E}$, the  graph-based Laplace mechanism is
$\mathcal{M}_{L}(w,f_{G},\epsilon)=f_{G}(w) + (Y_{1}, \ldots,Y_{k})$
where $Y_{i}$ are i.i.d. random variables drawn from Lap$(\Delta f_{G}/\epsilon)$, and Lap$(b)$ denotes the Laplace distribution with scale b $\left(\textit{{\em i.e} probability density }\frac{1}{2b}\exp\left(-\frac{|x|}{b}\right)\right)$.
\end{definition}

\begin{theorem}[\cite{Dwork_2006}]
\label{thm:Laplace}
The Laplace mechanism is $\epsilon$-differentially private.
\end{theorem}

We define hereafter the graph-based Exponential mechanism. In the sequel we refer to it simply as Exponential mechanism.
The Exponential mechanism represents a way of privately answering arbitrary range queries. Given some range of possible responses to the query $\mathcal{R}$, it is defined according to a utility function $u_{G}:= \mathcal{W}_{E} \times \mathcal{R} \rightarrow \mathbb{R}$, which aims at providing some total preorder on the range $\mathcal{R}$ according to the total order in $\mathbb{R}$. 
The sensitivity of this function is denoted $\Delta u_{G} := \max\limits_{r \in \mathcal{R}} \max\limits_{ w \sim w' \in \mathcal{W}_{E}} | u_{G}(w,r) - u_{G}(w',r)|$ .

\begin{definition}
Given some graph topology $G=(V,E)$, some output range $\mathcal{R} \subset E$, some privacy parameter $\epsilon > 0$, some utility function $u_{G}:= \mathcal{W}_{E} \times \mathcal{R} \rightarrow \mathbb{R}$, and some $w\in \mathcal{W}_{E}$ the graph-based Exponential mechanism $\mathcal{M}_{Exp}\left(G, w, u_{G}, \mathcal{R}, \epsilon \right)$ selects and outputs an element $r \in \mathcal{R}$ with probability proportional to $\exp\left(\frac{\epsilon u_{G}(w,r)}{2 \Delta u_{G}}\right)$.
\end{definition}

The Exponential mechanism defines a distribution on a potentially complex and large range $\mathcal{R}$. As the following theorem states,  sampling from such a distribution preserves $\epsilon$-differential privacy. 

\begin{theorem}[reformulation \cite{mechanism-design-via-differential-privacy}]
\label{th:privacyexpmechanism}
For any non-empty range $\mathcal{R}$, given some graph topology $G=(V,E)$, the graph-based Exponential mechanism preserves $\epsilon$-differential privacy, \emph{i.e} if $w \sim w' \in \mathcal{W}_{E}$,
\begin{align*}
\mathbb{P} &\left[\mathcal{M}_{Exp} \left(G,w, u_{G}, \mathcal{R}, \epsilon\right) = r\right] \\ 
&\leq e^{\epsilon}\mathbb{P}\left[\mathcal{M}_{Exp}\left(G,w', u_{G}, \mathcal{R}, \epsilon \right) =r\right].
\end{align*}
\end{theorem}

Further, Th~\ref{tradeoffexponential} highlights the trade-off between privacy and accuracy for the Exponential mechanism when $0 < |\mathcal{R}| <  +\infty$. Th~\ref{thm:basic-composition} presents the ability of differential privacy to comply with composition while Th~\ref{thm:postprocessing} introduces its post-processing property.

\begin{theorem}[reformulation \cite{Dwork_2013}]
\label{tradeoffexponential}
Given some graph topology $G=(V,E)$, some $w\in \mathcal{W}_{E}$, some output range $\mathcal{R}$, some privacy parameter $\epsilon > 0$, some utility function $u_{G}:= \mathcal{W}_{E} \times \mathcal{R} \rightarrow \mathbb{R}$, and denoting $OPT_{u_{G}}(w)=\max\limits_{r \in \mathcal{R}}u_{G}(w,r)$, one has $\forall$ $ t \in \mathbb{R}$, 
\begin{align*}
u_{G}&\left(G,w,\mathcal{M}_{Exp}\left(w,u_{G},\mathcal{R},\epsilon \right)\right) \\
&\leq OPT_{u_{G}}(w) - \frac{2\Delta u_{G}}{\epsilon}\left(t + \ln|\mathcal{R}|\right)
\end{align*}
with probability at most $\exp(-t)$.
\end{theorem}

\begin{theorem}[\cite{dwork2006our}] 
\label{thm:basic-composition}
For any $\epsilon >0$, $\delta \geq 0$ the adaptive composition of $k$ $(\epsilon,\delta)$-differentially private mechanisms is $(k\epsilon,k\delta)$-differentially private.
\end{theorem}

\begin{theorem}[Post-Processing~\cite{Dwork_2013}] 
\label{thm:postprocessing}
Let $\mathcal{A}:\mathcal{W}_{E} \rightarrow B$ be a randomized algorithm that is $(\epsilon,\delta)$-differentially private,  and $h : B \rightarrow B'$ a deterministic mapping.
Then $h \circ \mathcal{A}$ is $(\epsilon,\delta)$-differentially private.
\end{theorem}

\subsection{Differentially-private clustering}
Differentially private clustering for unstructured datasets has been first discussed in~\cite{Nissim_2007}. This work introduced the first method for differentially private clustering based on the k-means algorithm. Since then most of the works of the field focused on adaptation of this  method~\citep{Blum:2008,McSherry_2009,Dwork:2011}. The main drawback of those works is that they are not able to deal with arbitrary shaped clusters. This issue has been  recently investigated in~\cite{Ho:2013} and~\cite{Chen:2015}. They proposed two new methods to find arbitrary shaped clusters in unstructured datasets respectively based on density clustering and wavelet decomposition. Even though both of these works allow one to produce non-convex clusters, they only deal with unstructured datasets and thus are not applicable to node clustering in a graph. Our work focuses on node clustering in a graph under weight-differential privacy. Graph clustering has already been investigated in a topology-based privacy framework~\citep{Mlle2015,Nguyen:2016}, however, these works do not consider weight-differential privacy. Our work is, to the best of our knowledge, the first attempt to define node clustering in a graph under weight differential privacy.   

\section{Differentially-private tree-based clustering}

We aim at producing a private clustering method while providing bounds on the accuracy loss. Our method is an adaptation of an existing clustering algorithm \textsc{DBMSTClu}. However, to provide theoretical guarantees under differential privacy, one needs to rely on the same kind of guarantees in the non-private setting. \cite{MorvanCGA17} did not bring them in their initial work. Hence, our second contribution is to demonstrate the accuracy of this method, first in the non-private context.

In the following we present 1) the theoretical framework motivating MST-based clustering methods, 2) accuracy guarantees of \textsc{DBMSTClu} in the non-private setting, 3) \textsc{PTClust} our private clustering algorithm, 4) its accuracy under differential privacy constraints.  

\subsection{Theoretical framework for MST-based clustering methods} \label{sec:MST_clustering_optimality}

MST-based clustering methods, however efficient, lack of proper motivation. This Section closes this gap by providing a theoretical framework for MST-based clustering. 
In the sequel, notations from Section~\ref{sec:notations} are kept.
The minimum path distance between two nodes in the graph is defined which enables to explicit our notion of Cluster.

\begin{definition}[Minimum path distance]
Let be $\mathcal{G} = (V, E, w)$ and $u$, $v$ $\in V$. The minimum path distance between $u$ and $v$ is $$d(u,v) = \min_{\mathcal{P}_{u-v}} \underset{e \in V_{\mathcal{P}_{u-v}}} \sum w(e)$$
with $\mathcal{P}_{u-v}$ a path from $u$ to $v$ in $\mathcal{G}$, and $V_{\mathcal{P}_{u-v}}$ the set of vertices contained in $\mathcal{P}_{u-v}$.
\end{definition}

\begin{definition}[Cluster] \label{def:cluster}
Let be $\mathcal{G} = (V, E, w)$, $0 < w(e) \leq 1 \ \forall e \in E$ a graph, $(V, d)$ a metric space based on the minimum path distance $d$ defined on $\mathcal{G}$ and $D \subset V$ a node set.
$C \subset D$ is a cluster iff. $|C| > 2$ and $\forall C_1, C_2$ s.t. $C = C_1 \cup C_2$ and $C_1 \cap C_2 = \emptyset$, one has:
$$
\underset{z \in D \backslash C_1}{\argmin} \{ \ \underset{v \in C_1}{\min} \ d(z,v) \ \} \ \subset C_2
$$  
\end{definition}

Assuming that a cluster is built of at least $3$ points makes sense since singletons or groups of $2$ nodes can be legitimately considered as noise.  For simplicity of the proofs, the following theorems hold in the case where noise is neglected. However, they are still valid in the setting where noise is considered as singletons (with each singleton representing a generalized notion of cluster).

\begin{theorem}
\label{thm:MSTuseful}
Let be $\mathcal{G} = (V, E, w)$ a graph and $\mathcal{T}$ a minimum spanning tree of $\mathcal{G}$. 
Let also be $C$ a cluster in the sense of Def.~\ref{def:cluster} and
two vertices $v_1, v_2 \in C$.
Then, $V_{\mathcal{P}_{v_1-v_2}} \subset C$ 
with $\mathcal{P}_{v_1-v_2}$ a path from $v_1$ to $v_2$ in $\mathcal{G}$, and $V_{\mathcal{P}_{v_1-v_2}}$ the set of vertices contained in $\mathcal{P}_{v_1-v_2}$.
\end{theorem}


\begin{proof}
Let be $v_1, v_2 \in C$. 
If $v_1$ and $v_2$ are neighbors, the result is trivial.
Otherwise, as $\mathcal{T}$ is a tree, there exist a unique path within $\mathcal{T}$ between $v_1$ and $v_2$ denoted by $\mathcal{P}_{v_1-v_2} = \{ v_1, \ldots, v_2 \}$.
Let now prove by \textit{reductio ad absurdum} that $V_{\mathcal{P}_{v_1-v_2}} \subset C$.
Suppose there is $h \in V_{\mathcal{P}_{v_1-v_2}}$ s.t. $h \notin C$. We will see that it leads to a contradiction.
We set $C_1$ to be the largest connected component (regarding the number of vertices) of $\mathcal{T}$ s.t. $v_1 \in C_1$, and every nodes from $C_1$ are in $C$. Because of $h$'s definition, $v_2 \notin C_1$. 
Let be $C_2 = C \backslash C_1$. $C_2 \neq \emptyset$ since $v_2 \in C_2$.
Let be $z^* \ \in \ \underset{z \in V \backslash C_1}{\argmin} \{ \ \underset{v \in C_1}{\min} \ d(z, v) \ \}$ and $e^*= (z^*,v^*)$ an edge that reaches this minimum. Let us show that $z^* \notin C$.
If $z^* \in C$, then two possibilities hold:
\begin{enumerate}
\item There is an edge $e_{z^*} \in \mathcal{T}$, s.t. $e_{z^*} = (z^*, z')$ with $z' \in C_1$. This is impossible, otherwise by definition of a connected component, $z^* \in C_1$. \underline{Contradiction}.
\item For all $e_{z^*} = (z^*, z')$ s.t $z' \in C_1$, one has $e_{z^*} \notin \mathcal{T}$. In particular $e^* \notin \mathcal{T}$.
Since $h$ is the neighbor of $C_1$ in $\mathcal{G}$ there is also $e_h \in \mathcal{T}$, s.t. $e_h = (h, h')$ with $h' \in C_1$.
Once again two possibilities hold:
\begin{enumerate}
\item $w(e_{z^*}) = \underset{z \in V \backslash C_1}{\min} \{ \ \underset{v \in C_1}{\min} \ d(z, v) \ \} < w(e_h)$. Then, if we replace $e_h$ by $e_{z^*}$ in $\mathcal{T}$, its total weight decreases. So $\mathcal{T}$ is not a minimum spanning tree. \underline{Contradiction}.
\item $w(e_{z^*}) = w(e_h)$, therefore $h \in  \underset{z \in V \backslash C_1}{\argmin} \{ \ \underset{v \in C_1}{\min} \ d(z, v) \ \}$. Since $h \notin C$, one gets that \\ $\underset{z \in V \backslash C_1}{\argmin} \{ \ \underset{v \in C_1}{\min} \ d(z,v) \ \} \ \not \subset C_2$. Thus, $C$ is not a cluster. \underline{Contradiction}.
\end{enumerate}
\end{enumerate}
We proved that $z^* \notin C$. In particular, $z^* \notin C_2$. Then, $\underset{z \in V \backslash C_1}{\argmin} \{ \ \underset{\ v \in C_1}{\min} \ d(z, v) \ \} \ \not \subset C_2$. Thus, $C$ is not a cluster. \underline{Contradiction}.
Finally $h \in C$ and $V_{\mathcal{P}_{v_1-v_2}} \subset C$.
\end{proof}

This theorem states that, given a graph $\mathcal{G}$,  an MST $\mathcal{T}$, and any two nodes of $C$, every node in the path between them is in $C$. This means that a cluster can be characterized by a subtree of $\mathcal{T}$.
It justifies the use of all MST-based methods for data clustering or node clustering in a graph. All the clustering algorithms based on successively cutting edges in an MST to obtain a subtree forest are meaningful in the sense of Th.\ref{thm:MSTuseful}.
In particular, this theorem holds for the use of \textsc{DBMSTClu}~\citep{MorvanCGA17} presented in Section~\ref{sec:DBMSTClu}. 

\subsection{Deterministic MST-based clustering}

This Section introduces \textsc{DBMSTClu}~\citep{MorvanCGA17} that will be adapted to be differentially-private, and provide accuracy results on the recovery of the ground-truth clustering partition. 

\subsubsection{\textsc{DBMSTClu} algorithm} \label{sec:DBMSTClu}
Let us consider $\mathcal{T}$ an MST of $\mathcal{G}$, as the unique input of the clustering algorithm \textsc{DBMSTClu}. The clustering partition results then from successive cuts on $\mathcal{T}$ so that a new cut in $\mathcal{T}$ splits a connected component into two new ones. Each final connected component, a subtree of $\mathcal{T}$, represents a cluster. Initially, $\mathcal{T}$ is one cluster containing all nodes. Then, at each iteration, an edge is cut if some criterion, called \emph{Validity Index of a Clustering Partition} (DBCVI) is improved. This edge is greedily chosen to locally maximize the DBCVI at each step. When no improvement on DBCVI can be further made, the algorithm stops.
The DBCVI is defined as the weighted average of all \emph{cluster validity indices} which are based on two positive quantities, the \textit{Dispersion} and the \textit{Separation} of a cluster:
\begin{definition}[Cluster Dispersion]
The Dispersion of a cluster $C_i$ ($\DISP$) is defined as the maximum edge weight of $C_i$. If the cluster is a singleton (i.e. contains only one node), the associated Dispersion is set to $0$. More formally:
\begin{equation*}
\forall i \in [K], \ \DISP(C_i) = \left\{
    \begin{array}{ll}
        \max\limits_{j, \ e_j \in C_i} w_j & \mbox{if } \ |E(C_i)| \neq 0 \\
        0 & \mbox{otherwise.}
    \end{array}
\right.
\end{equation*}
\end{definition}
\begin{definition}[Cluster Separation]
The Separation of a cluster $C_i$ ($\SEP$) is defined as the minimum distance between the nodes of $C_i$ and the ones of all other clusters $C_j, j \neq i, 1 \leq i,j \leq K, K \neq 1$ where $K$ is the total number of clusters. In practice, it corresponds to the minimum weight among all already cut edges from $\mathcal{T}$ comprising a node from $C_i$. If $K = 1$, the Separation is set to $1$. More formally, 
with $incCuts(C_i)$ denoting cut edges incident to $C_i$,
\begin{equation*}
\forall i \in [K], \ \SEP(C_i) = \left\{
    \begin{array}{ll}
        \min\limits_{j, \ e_j \in incCuts(C_i)} w_j & \mbox{if } \ K \neq 1 \\
         1 & \mbox{otherwise.}
    \end{array}
\right.
\end{equation*}
\end{definition}

\begin{definition}[Validity Index of a Cluster]
The Validity Index of a cluster $C_i$ is defined as:
\begin{equation*}
V_C(C_i) = \frac{ \SEP(C_i) - \DISP(C_i)}{ \max(\SEP(C_i), \DISP(C_i))} \in [-1; 1]
\end{equation*}
\end{definition}

\begin{definition}[Validity Index of a Clustering Partition]
The Density-Based Validity Index of a Clustering partition $\Pi = \{C_i\}, 1 \leq i \leq K$, $\DBCVI(\Pi)$ is defined as the weighted average of the Validity Indices of all clusters in the partition where $N$ is the number of vertices.
\vspace*{-0.1in}
\begin{equation*}
\DBCVI(\Pi) = \sum_{i = 1}^{K} \frac{ | C_i| }{N} V_C(C_i) \in [-1,1]
\end{equation*}
\end{definition}

\textsc{DBMSTClu} is summarized in Algorithm~\ref{alg:DBMSTClu}: $\evaluateCut(.)$ computes the DBCVI when the cut in parameter is applied to $\mathcal{T}$. Initial DBCVI is set $-1$. Interested reader could refer to~\citep{MorvanCGA17} Section 4. for a complete insight on this notions. 
\begin{algorithm}[!htbp]
\caption{\textsc{DBMSTClu}($\mathcal{T}$)\label{alg:DBMSTClu}}
\begin{algorithmic}[1]
\STATE {\bfseries Input:} $\mathcal{T}$, the MST
\STATE $dbcvi \leftarrow -1.0$
\STATE $clusters \leftarrow \emptyset$
\STATE $cut\_list \leftarrow \{E(\mathcal{T}) \}$
\WHILE{ $dbcvi < 1.0$ }
\STATE $cut\_tp \leftarrow \emptyset$
\STATE $dbcvi\_tp \leftarrow dbcvi$
\FOR{each $cut$ in $cut\_list$}
\STATE $newDbcvi \leftarrow \evaluateCut(\mathcal{T}, cut)$
\IF{$newDbcvi \geq dbcvi\_tp$} 
\STATE $cut\_tp \leftarrow cut$
\STATE $dbcvi\_tp \leftarrow  newDbcvi$
\ENDIF
\ENDFOR
\IF{$cut\_tp \neq \emptyset$}
\STATE $clusters \small{\leftarrow} \cut(clusters, cut\_tp)$
\STATE $dbcvi \leftarrow dbcvi\_tp$
\STATE $cut\_list  \leftarrow cut\_list \backslash \{cut\_tp \}$
\ELSE 
\STATE {\bfseries break}
\ENDIF
\ENDWHILE
\RETURN $clusters$, $dbcvi$
\end{algorithmic}
\end{algorithm}

\subsubsection{DBMSTClu exact clustering recovery proof} \label{sec:DBMSTClu_proofs}
In this section, we provide theoretical guarantees for the cluster recovery accuracy of DBMSTClu. Let us first begin by introducing some definitions.

\begin{definition}[Cut] \label{def:cut}
Let us consider a graph $\mathcal{G} = (V, E, w)$ with $K$ clusters, $\mathcal{T}$ an MST of $\mathcal{G}$.
Let denote $(C^*_i)_{i \in [K]}$ the set of the clusters.
Then, $Cut_{\mathcal{G}}(\mathcal{T}) := \{ e_{kl} \in \mathcal{T} \ | \ k \in C^*_i, \ l \in C^*_j, \ i, j \in [K]^2, \ i \neq j \}$. In the sequel, for simplicity, we denote $e^{(ij)} \in Cut_{\mathcal{G}}(\mathcal{T})$ the edge between cluster $C^*_i$ and $C^*_j$.
\end{definition}

$Cut_{\mathcal{G}}(\mathcal{T})$ is basically the set of effective cuts to perform on $\mathcal{T}$ in order to ensure the exact recovery of the clustering partition. More generally, trees on which $Cut_{\mathcal{G}}(.)$ enables to find the right partition are said to be a partitioning topology.

\begin{definition}[Partitionning topology] \label{def:partition_topo}
Let us consider a graph $\mathcal{G} = (V, E, w)$ with $K$ clusters $C^{*}_1, \ldots, C^{*}_K$.
A spanning tree $\mathcal{T}$ of $\mathcal{G}$ is said to have a partitioning topology if $\forall i, j \in [K]$, $i \neq j$, $| \{ e=(u,v) \in Cut_{\mathcal{G}}(\mathcal{T}) \ | \ u \in C^*_i,  v \in C^*_j \} | = 1$.
\end{definition}

Def.~\ref{def:cut} and~\ref{def:partition_topo} introduce a topological condition on the tree as input of the algorithm. Nevertheless, conditions on weights are necessary too. Hence, we define homogeneous separability which expresses the fact that within a cluster the edge weights are spread in a controlled manner.

\begin{definition}[Homogeneous separability condition] \label{def:homo_separ_cond}
Let us consider a graph $\mathcal{G} = (V, E, w)$, $s \in E$ and $\mathcal{T}$ a tree of $\mathcal{G}$. $\mathcal{T}$ is said to be homogeneously separable by $s$, if $$\alpha_\mathcal{T} \ \underset{e \in E(\mathcal{T})}{\max} \ w(e) < w(s)
\textnormal{ with } \alpha_{\mathcal{T}} = \frac{ \underset{e \in E(\mathcal{T})}{\max} \ w(e) }{ \underset{e \in E(\mathcal{T})}{\min} \ w(e) } \geq 1.$$ One will write for simplicity that $H_{\mathcal{T}}(s)$ is verified.
\end{definition}



\begin{definition}[Weak homogeneity condition of a Cluster] \label{def:cond}
Let us consider a graph $\mathcal{G} = (V, E, w)$ with $K$ clusters $C^*_1, \ldots, C^*_K$.
A given cluster $C^*_i$, $i \in [K]$, $C^*_i$ is weakly homogeneous if:
for all $ \mathcal{T}$ an MST of  $\mathcal{G}$, and 
  $\forall j \in [K]$, $j \neq i$, s.t. $e^{(ij)} \in Cut_{\mathcal{G}}(\mathcal{T}), \ H_{\mathcal{T}_{|C_i^*}}(e^{(ij)})$ is verified.
For simplicity, one denote $\ubar{\alpha}_i = \underset{\mathcal{T} \textnormal{ MST of } \mathcal{G}}{\max}\alpha_{\mathcal{T}_{|C_i^*}}$
\end{definition}

\begin{definition}[Strong homogeneity condition of a Cluster] \label{def:cond}
Let us consider a graph $\mathcal{G} = (V, E, w)$ with $K$ clusters $C^*_1, \ldots, C^*_K$.
A given cluster $C^*_i$, $i \in [K]$, $C^*_i$ is strongly homogeneous if:
for all $ \mathcal{T}$ a spanning tree (ST) of  $\mathcal{G}$, and 
  $\forall j \in [K]$, $j \neq i$, s.t. $e^{(ij)} \in Cut_{\mathcal{G}}(\mathcal{T}), \ H_{\mathcal{T}_{|C_i^*}}(e^{(ij)})$ is verified.
For simplicity, one denote $\bar{\alpha}_i = \underset{\mathcal{T} \textnormal{ ST of } \mathcal{G}}{\max}\alpha_{\mathcal{T}_{|C_i^*}}$
\end{definition}
We show that the weak homogeneity condition is implied by the strong homogeneity condition.
\begin{proposition} \label{prop:link_strong_weak}
Let us consider a graph $\mathcal{G} = (V, E, w)$ with $K$ clusters $C^{*}_1, \ldots, C^{*}_K$.
If a given cluster $C^{*}_i$, $i \in [K]$ is strongly homogeneous, then, it is weakly homogeneous.
\end{proposition}
\begin{proof}
If $ \mathcal{T}$ a spanning tree of  $\mathcal{G}$, and 
  $\forall j \in [K]$, $j \neq i$, s.t. $e^{(ij)} \in Cut_{\mathcal{G}}(\mathcal{T}), \ H_{\mathcal{T}_{|C_i^*}}(e^{(ij)})$ is verified, then in particular, it is true for any MST.
\end{proof}

Strong homogeneity condition appears to be naturally more constraining on the edge weights than the weak one. The  accuracy of \textsc{DBMSTClu} is proved under the weak homogeneity condition, while the accuracy of its differentially-private version is only given under the   the strong homogeneity condition.


\begin{theorem} \label{thm:thmA} 
Let us consider a graph $\mathcal{G} = (V, E, w)$ with $K$ homogeneous clusters $C^*_1, \ldots, C^*_K$ and $\mathcal{T}$ an MST of $\mathcal{G}$.
Let now assume that at step $k < K-1$, DBMSTClu built $k+1$ subtrees $\mathcal{C}_1, \ldots, \mathcal{C}_{k+1}$ by cutting $e_1, \ e_2, \ \ldots, \ e_k \in E$.

Then, $Cut_k := Cut_{\mathcal{G}}(\mathcal{T}) \ \backslash \ \{ e_1, \ e_2, \ \ldots, \ e_k \} \neq \emptyset \implies \DBCVI_{k+1} \geq DBCVI_k$, i.e. if there are still edges in $Cut_k$, the algorithm will continue to perform some cut.
\end{theorem}
\begin{proof}
See supplementary material.
\end{proof}

\begin{theorem} \label{thm:thmB}
Let us consider a graph $\mathcal{G} = (V, E, w)$ with $K$ homogeneous clusters $C^*_1, \ldots, C^*_K$ and $\mathcal{T}$ an MST of $\mathcal{G}$.

Assume now that at step $k < K - 1$, DBMSTClu built $k+1$ subtrees $\mathcal{C}_1, \ldots, \mathcal{C}_{k+1}$ by cutting \\
$e_1, \ e_2, \ \ldots, \ e_k \in E$. We still denote $Cut_k := Cut_{\mathcal{G}}(\mathcal{T}) \backslash \{ e_1, \ e_2, \ \ldots, \ e_k \}$.

If $Cut_k \neq \emptyset$ then $\underset{e \in \mathcal{T} \backslash \{ e_1, \ e_2, \ \ldots, \ e_k \}}{\argmax} DBCVI_{k+1}(e) \subset Cut_k$ i.e. the cut edge at step $k+1$ is in $Cut_k$. 
\end{theorem}
\begin{proof}
See supplementary material.
\end{proof}

\begin{theorem} \label{thm:thmC} 
Let us consider a graph $\mathcal{G} = (V, E, w)$ with $K$ weakly homogeneous clusters $C^*_1, \ldots, C^*_K$ and $\mathcal{T}$ an MST of $\mathcal{G}$.
Let now assume that at step $K-1$, DBMSTClu built $K$ subtrees $\mathcal{C}_1, \ldots, \mathcal{C}_{K}$ by cutting $e_1, \ e_2, \ \ldots, \ e_{K-1} \in E$. We still denote $Cut_{K-1} := Cut_{\mathcal{G}}(\mathcal{T}) \backslash \{ e_1, \ e_2, \ \ldots, \ e_{K-1} \}$.
 
Then, for all $e \in \mathcal{T} \backslash \{ e_1, \ e_2, \ \ldots, \ e_{K-1} \}$, $DBCVI_{K}(e) < DBCVI_{K-1}$ i.e. the algorithm stops: no edge gets cut during step $K$.
\end{theorem}
\begin{proof}
See supplementary material.
\end{proof}

\begin{corollary}
Let us consider a graph $\mathcal{G} = (V, E, w)$ with $K$ weakly homogeneous clusters $C^*_1, \ldots, C^*_K$ and $\mathcal{T}$ an MST of $\mathcal{G}$.
$DBMSTClu(\mathcal{T})$ stops after $K-1$ iterations and the $K$ subtrees produced match exactly the clusters i.e. under homogeneity condition, the algorithm finds automatically the underlying clustering partition.
\end{corollary}
\begin{proof}
Th.~\ref{thm:thmA} and~\ref{thm:thmC} ensure that under homogeneity condition on all clusters, the algorithm performs the $K-1$ distinct cuts within $Cut_{\mathcal{G}}(\mathcal{T})$ and stops afterwards.  
By definition of $Cut_{\mathcal{G}}(\mathcal{T})$, it means the DBMSTClu correctly builds the $K$ clusters. 
\end{proof}

\subsection{Private MST-based clustering} \label{sec:private_mst_clustering}

This section presents our new node clustering algorithm \textsc{PTClust} for weight differential privacy. It relies on a mixed adaptation of \textsc{PAMST} algorithm~\citep{Pinot2017} for recovering a differentially-private MST of a graph and \textsc{DBMSTClu}.

\subsubsection{PAMST algorithm}
Given a simple-undirected-weighted graph $\mathcal{G}=(V,E,w)$, \textsc{PAMST} outputs an almost minimal weight spanning tree topology under differential privacy constraints. 
It relies on a Prim-like MST algorithm, and an iterative use of the graph based Exponential mechanism. \textsc{PAMST} takes as an input a weighted graph, and a utility function. It outputs the topology of a spanning tree which weight is almost minimal. Algorithm~\ref{alg:Private-MST} presents this new method, using the following utility function: 
$$ \begin{array}{ccccc}
u_{G} & : & \mathcal{W}_{E} \times \mathcal{R} & \to & \mathbb{R} \\
& & (w,r) & \mapsto & - |w(r) - \min\limits_{r' \in \mathcal{R}} w(r')|. \end{array}$$
\textsc{PAMST} starts by choosing an arbitrary node to construct iteratively the tree topology. At every iteration, it uses the Exponential mechanism to find the next edge to be added to the current tree topology while keeping the weights private. This algorithm is the state of the art to find a spanning tree topology under differential privacy.
For readability, let us introduce some additional notations. Let $S$ be a set of nodes from $G$, and $\mathcal{R}_{S}$ the set of edges that are incident to one and only one node in $S$ (also denoted xor-incident). For any edge $r$ in such a set, the incident node to $r$ that is not in $S$ is denoted $r_{\rightarrow}$.
Finally, the restriction of the weight function to an edge set $\mathcal{R}$ is denoted $w_{|\mathcal{R}}$.
\begin{algorithm}[!htbp]
\caption{\textsc{PAMST}$(G,u_{G},w,\epsilon)$}
\label{alg:Private-MST}
\begin{algorithmic}[1] 
\STATE {\bfseries Input:} $\mathcal{G}=(V,E,w)$ a weighted graph (separately the topology $G$ and the weight function $w$), $\epsilon$ a degree of privacy and $u_{G}$ utility function. 
\STATE {\bfseries Pick} $v \in V$ at random 
\STATE $S_{V} \leftarrow \{v\}$
\STATE $S_{E} \leftarrow \emptyset$
\WHILE{$S_{V}\neq V$ }
\STATE $r=\mathcal{M}_{Exp}(\mathcal{G},w,u_{G},\mathcal{R}_{S_{V}},\frac{\epsilon}{|V|-1})$
\STATE $S_{V} \leftarrow S_{V} \cup \{r_{\rightarrow}\}$
\STATE $S_{E} \leftarrow S_{E} \cup \{r\}$
\ENDWHILE
\RETURN $S_{E}$
\end{algorithmic}
\end{algorithm}

Theorem~\ref{thm:Privacy-PAMST} states that using PAMST to get an almost minimal spanning tree topology preserves weight-differential privacy.
\begin{theorem}
\label{thm:Privacy-PAMST}
Let $G=(V,E)$ be the topology of a simple-undirected graph, then $\forall \epsilon >0$, \textsc{PAMST}$\left(G,u_{G},\bullet,\epsilon\right)$  is $\epsilon$- differentially private on $G$. 
\end{theorem}

\subsubsection{Differentially private clustering}
The overall goal of this Section is to show that one can obtain a differentially private clustering algorithm by combining \textsc{PAMST} and \textsc{DBMSTClu} algorithms.
However, \textsc{PAMST} does not output a weighted tree which is inappropriate for clustering purposes. To overcome this, one cloud rely on a sanitizing mechanism such as the Laplace mechanism. Moreover, since \textsc{DBMSTClu} only takes weights from (0,1], two normalizing parameters $\tau$ and $p$ are introduced, respectively to ensure lower and upper bounds to the weights that fit within \textsc{DBMSTClu} needs. This sanitizing mechanism is called the Weight-Release mechanism.
Coupled with \textsc{PAMST}, it will allows us to produce a weighted spanning tree with differential privacy, that will be exploited in our private graph clustering.

\begin{definition}[Weight-Release mechanism]
Let $\mathcal{G}=(G,w)$ be a weighted graph, $\epsilon > 0$ a privacy parameter, $s$ a scaling parameter, $\tau \geq 0$, and $p \geq 1$ two normalization parameters. The Weight-Release mechanism is defined as
$$ \mathcal{M}_{w.r}(G,w,s,\tau,p)=\left(G,w'=\frac{w+ (Y_{1},...,Y_{|E|}) + \tau }{p}\right)$$ where $Y_{i}$ are i.i.d. random variables drawn from $Lap\left(0,s\right)$. 
With $w+ (Y_{1},...,Y_{|E|})$ meaning that if one gives an arbitrary order to the edges $E=(e_{i})_{i \in [|E|]}$, one has $\forall  i \in [|E|]$, $w'(e_{i})=w(e_{i}) + Y_{i}$.
\end{definition}
The following theorem presents the privacy guarantees of the Weight-Release mechanism. 

\begin{theorem}
\label{thm:relese-mechanism}
Let $G=(V,E)$ be the topology of a simple-undirected graph, $\tau\geq 0$, $p\geq 1$, then $\forall \epsilon >0$, $\mathcal{M}_{w.r}\left(G,\bullet,\frac{\mu}{\epsilon},\tau,p\right)$ is $\epsilon$- differentially private on $G$.  
\end{theorem}

\begin{proof}
Given $\tau \geq 0$, $p \geq 1$, and $\epsilon >0$, the Weight release mechanism scaled to $\frac{\mu}{\epsilon}$ can be break down into a Laplace mechanism and a post-processing consisting in adding $\tau$ to every edge and dividing them by $p$. Using Theorems~\ref{thm:Laplace} and~\ref{thm:postprocessing}, one gets the expected result.
\end{proof}

So far we have presented \textsc{DBMSTClu} and \textsc{PAMST} algorithms, and the Weight-Release mechanism. Let us now introduce how to compose those blocks to obtain a Private node clustering in a graph, called \textsc{PTClust}.
\begin{algorithm}[!htbp]
\caption{\textsc{PTClust}$(G,w,u_{G},\epsilon,\tau,p)$}
\label{alg:Private-MST}
\begin{algorithmic}[1] 
\STATE {\bfseries Input:} $\mathcal{G}=(V,E,w)$ a weighted graph (separately the topology $G$ and the weight function $w$), $\epsilon$ a degree of privacy and $u_{G}$ utility function. 
\STATE $T=\textsc{PAMST}(G,w,u_{G},\epsilon/2)$
\STATE $\mathcal{T'}=\mathcal{M}_{w.r}(T,w_{|E(T)},\frac{2\mu}{\epsilon},\tau,p)$ 

\RETURN \textsc{DBMSTClu}$(\mathcal{T'})$
\end{algorithmic}
\end{algorithm}
The algorithm takes as an input a weighted graph (dissociated topology and weight function), a utility function, a privacy degree and two normalization parameters. It outputs a clustering partition. To do so, a spanning tree topology is produced using \textsc{PAMST}. Afterward a randomized and normalized version of the associated weight function is released using the Weight-release mechanism. Finally the obtained weighted tree is given as an input to \textsc{DBMSTClu} that performs a clustering partition. The following theorem ensures that our method preserves $\epsilon$-differential privacy.

\begin{theorem}
Let $G=(V,E)$ be the topology of a simple-undirected graph, $\tau \geq 0$, and $p\geq 1$, then $\forall \epsilon >0$,
\textsc{PTClust}$(G,\bullet,u_{G},\epsilon,\tau,p)$ is $\epsilon$-differentially private on $G$.
\end{theorem}

\begin{proof}
Using Theorem~\ref{thm:Privacy-PAMST} one has that $T$ is produced with $\epsilon/2$-differential privacy, and using Theorem~\ref{thm:relese-mechanism} one has that $w'$ is obtained with $\epsilon/2$-differential privacy as well. Therefore using Theorem~\ref{thm:basic-composition}, $\mathcal{T'}$ is released with $\epsilon$-differential privacy. 
Using the post-processing property (Theorem~\ref{thm:postprocessing}) one gets the expected result.
\end{proof}

\subsection{Differential privacy trade-off of clustering} \label{sec:accuracy_under_privacy}
The results stated in this section present the security/accuracy trade-off of our new method in the differentially-private framework. \textsc{PTClust} relies on two differentially private mechanisms, namely \textsc{PASMT} and the Weight-Release mechanism. Evaluating the accuracy of this method amounts to check whether using these methods for ensuring privacy does not deteriorate the final clustering partition. The accuracy is preserved if \textsc{PAMST} outputs the same topology as the MST-based clustering, and if the Weight-Release mechanism preserves enough the weight function.
According to Def.~\ref{def:partition_topo}, if a tree has a partitioning topology, then it fits the tree-based clustering. The following theorem states that with high probability \textsc{PAMST} outputs a tree with a partitioning topology.  

\begin{theorem} \label{thm:thmD}
Let us consider a graph $\mathcal{G} = (V, E, w)$ with $K$ strongly homogeneous clusters $C^{*}_1, \ldots, C^{*}_K$ and $T = \textsc{PAMST}( \mathcal{G}, u_{\mathcal{G}}, w, \epsilon)$, $\epsilon > 0$. $T$ has a partitioning topology with probability at least
\begin{align*}
1 -
\sum^{K}_{i = 1} \left( |C^*_i| - 1 \right) \exp\left( -\frac{A}{2 \Delta u_{\mathcal{G}} (|V| - 1)} \right) 
\end{align*}
with $A = \epsilon \left( \bar{\alpha}_i \underset{e \in  E\left(\mathcal{G}_{|C^*_i}\right) }{\max( w(e))} - \underset{e \in  E\left(\mathcal{G}_{|C^*_i}\right) }{\min{(w(e))}} \right) + \ln|E|$.
\end{theorem}
\begin{proof}
See supplementary material.
\end{proof}

The following theorem states that given a tree $\mathcal{T}$ under the strong homogeneity condition, if the subtree associated to a cluster respects Def.~\ref{def:homo_separ_cond}, then it still holds after applying the Weight-Release mechanism to this tree. 


\begin{theorem}\label{thm:tradeoff}
Let us consider a graph $\mathcal{G} = (V, E, w)$ with $K$ strongly homogeneous clusters $C^{*}_1, \ldots, C^{*}_K$ and $T = PAMST( \mathcal{G}, u_{\mathcal{G}}, w, \epsilon)$, $\mathcal{T}=(T,w_{|T})$ and $\mathcal{T'}=\mathcal{M}_{w.r}(T,w_{|T},s,\tau,p)$ with $s<<p,\tau$. Given some cluster $C^{*}_i$, and $j \neq i$ s.t $e^{(ij)} \in Cut_{\mathcal{G}}(\mathcal{T})$, if $H_{\mathcal{T}_{|C_i^*}}(e^{(ij)})$ is verified, then $H_{\mathcal{T'}_{|C_i^*}}(e^{(ij)})$ is verified with probability at least $$ 1 - \frac{\mathbb{V}(\varphi)}{ \mathbb{V}(\varphi) + \mathbb{E}(\varphi)^2} $$ 
with the following notations :
\begin{itemize}
\item $ \varphi =\underset{j \in [|C_i^*|-1] }{(\max Y_j)^2} - \underset{j \in [|C_i^*|-1] }{\min Z_j} \times X^{out} $
\item $Y_j\underset{iid}{\sim} Lap\left(\frac{\underset{e\in E(\mathcal{T})}{\max} w(e) + \tau}{p},\frac{s}{p}\right)$
\item $Z_j \underset{iid}{\sim} Lap\left( \frac{\underset{e\in E(\mathcal{T})}{\min} w(e) + \tau}{p}, \frac{s}{p}\right)$
\item $X^{out}\sim Lap\left(\frac{w(e^{(ij)} + \tau}{p}, \frac{s}{p}\right),$ 
\end{itemize}
\end{theorem}
\begin{proof}
See supplementary material.
\end{proof}
Note that Theorem~\ref{thm:tradeoff} is stated in a simplified version. A more complete version (specifying an analytic version of $\mathbb{V}(\varphi)$ and $\mathbb{E}(\varphi)$) is given in the supplementary material.

\section{Experiments}
\begin{figure*}[!t]
\vspace{1in}
\begin{subfigure}{0.21\textwidth} 
\includegraphics[width=0.9\columnwidth]{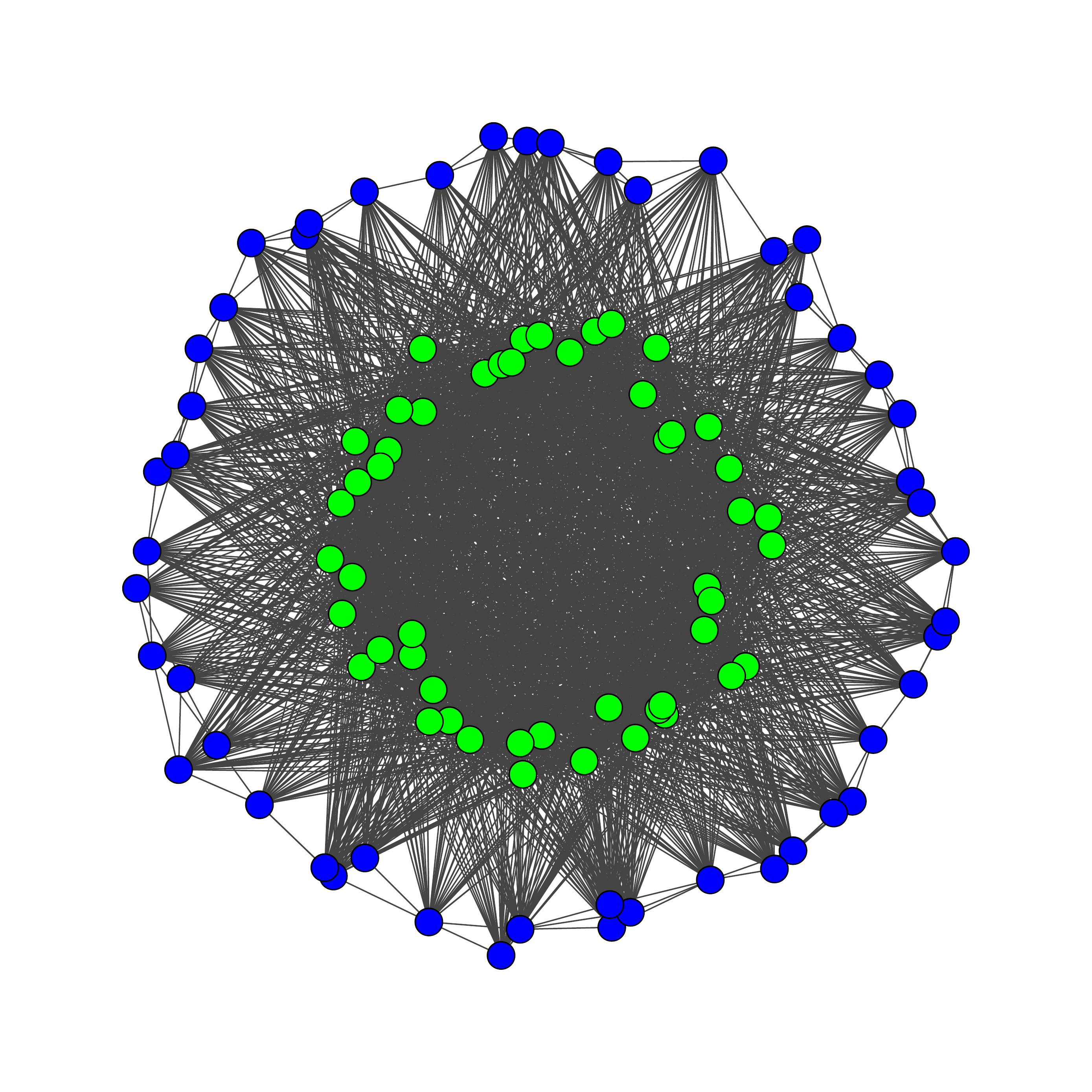}
\caption{Homogeneous graph}
\label{subfig:circles_homo}
\end{subfigure} \hspace*{-0.15in}
\begin{subfigure}{0.21\textwidth} 
\includegraphics[width=0.9\columnwidth]{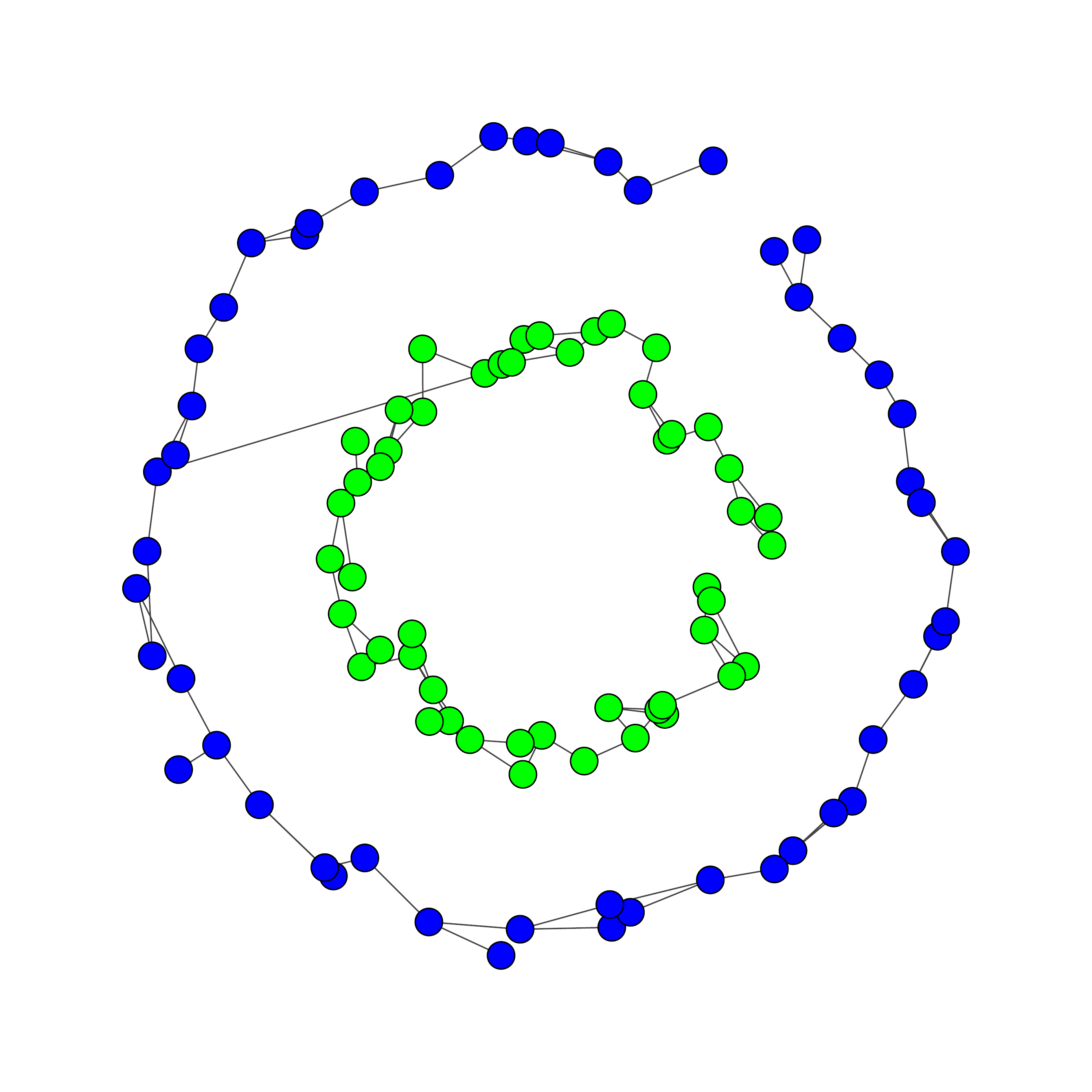}
\caption{\textsc{DBMSTClu} \label{subfig:circles_DBMSTClu}}
\end{subfigure} \hspace*{-0.15in}
\begin{subfigure}{0.21\textwidth} 
\includegraphics[width=0.9\columnwidth]{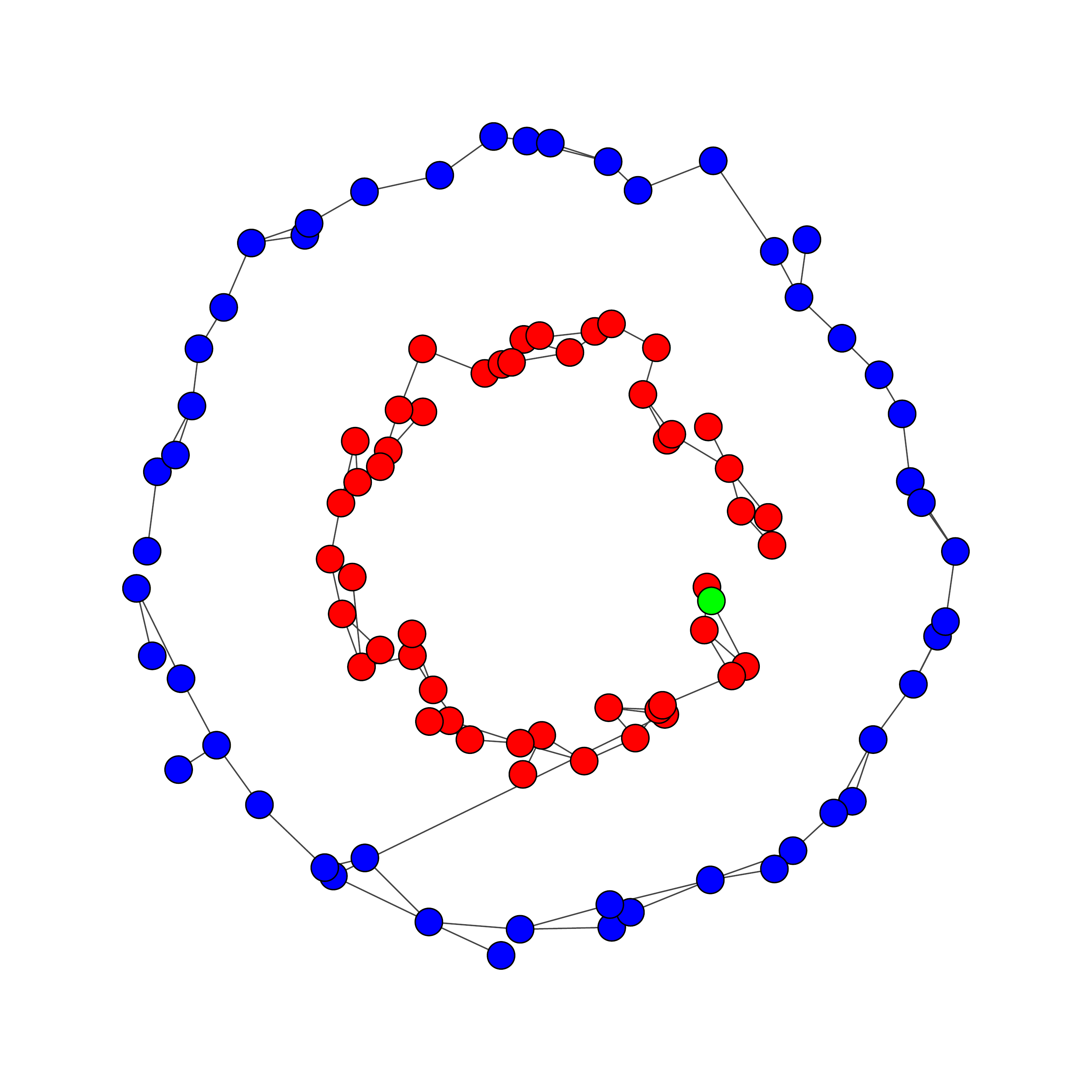}
\caption{\textsc{PTClust}, $\epsilon = 1.0$}
\label{subfig:circles_PTClu1}
\end{subfigure} \hspace*{-0.15in}
\begin{subfigure}{0.21\textwidth} 
\includegraphics[width=0.9\columnwidth]{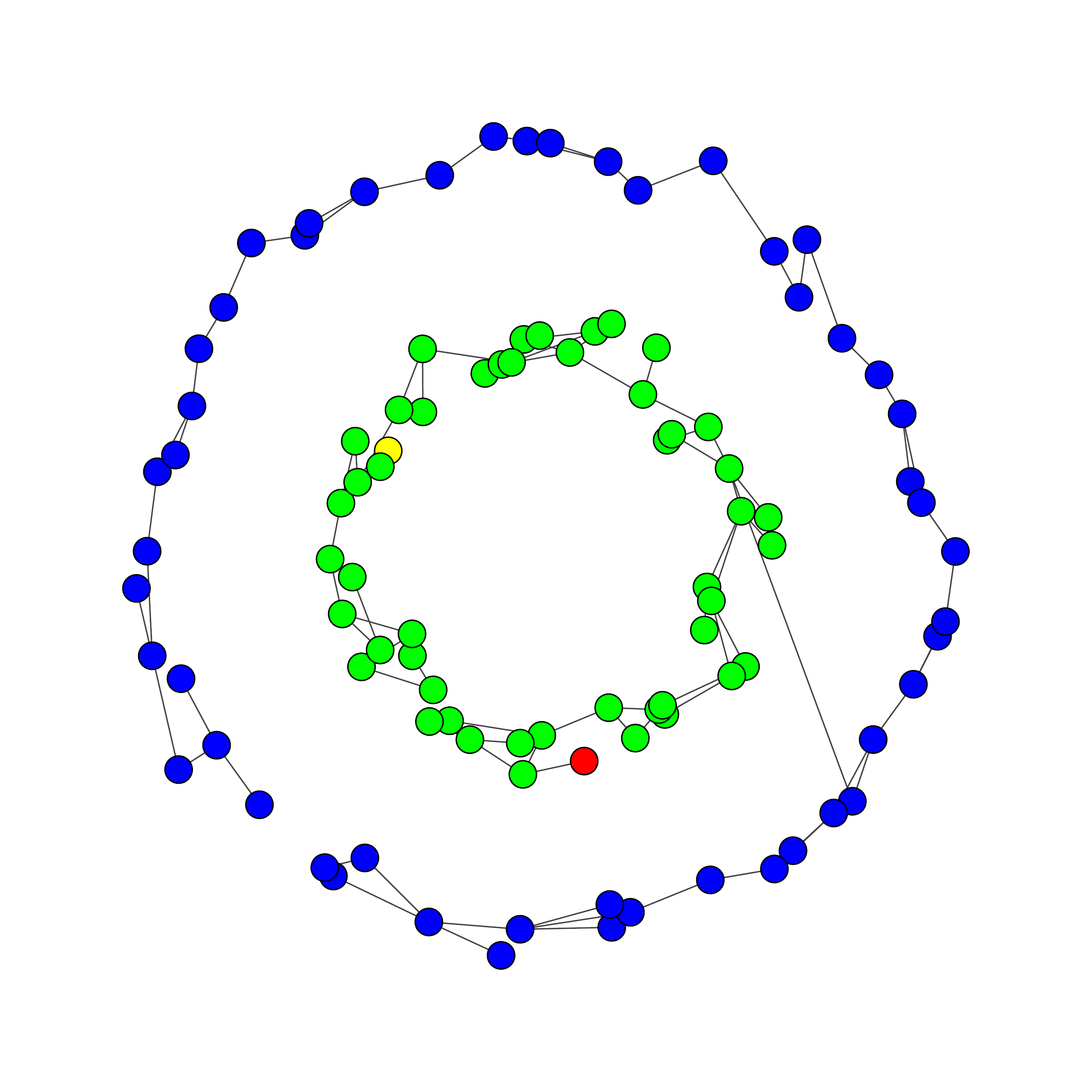}
\caption{\textsc{PTClust}, $\epsilon = 0.7$ \label{subfig:circles_PTClu07}}
\end{subfigure} \hspace*{-0.15in}
\begin{subfigure}{0.21\textwidth} 
\includegraphics[width=0.9\columnwidth]{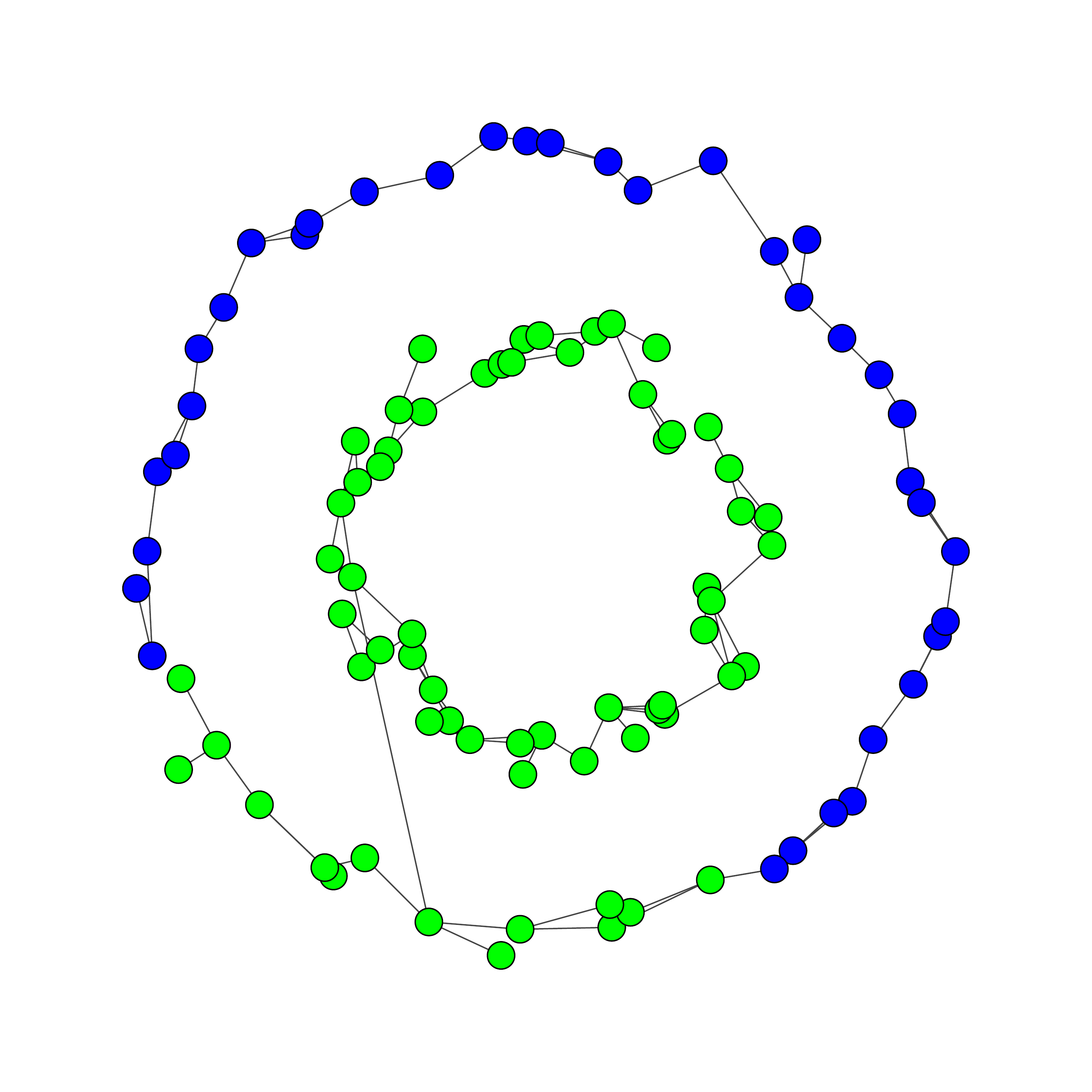}
\caption{\textsc{PTClust}, $\epsilon = 0.5$}
\label{subfig:circles_PTClu05}
\end{subfigure} \hspace*{-0.15in}
\caption{Circles experiments for $n = 100$. \textsc{PTClust} parameters: $w_{min} = 0.1$, $w_{max} = 0.3$, $\mu = 0.1$.}
\label{fig:circles}
\end{figure*}

\begin{figure*}[!t]
\vspace{1in}
\begin{subfigure}{0.21\textwidth} 
\includegraphics[width=0.9\columnwidth]{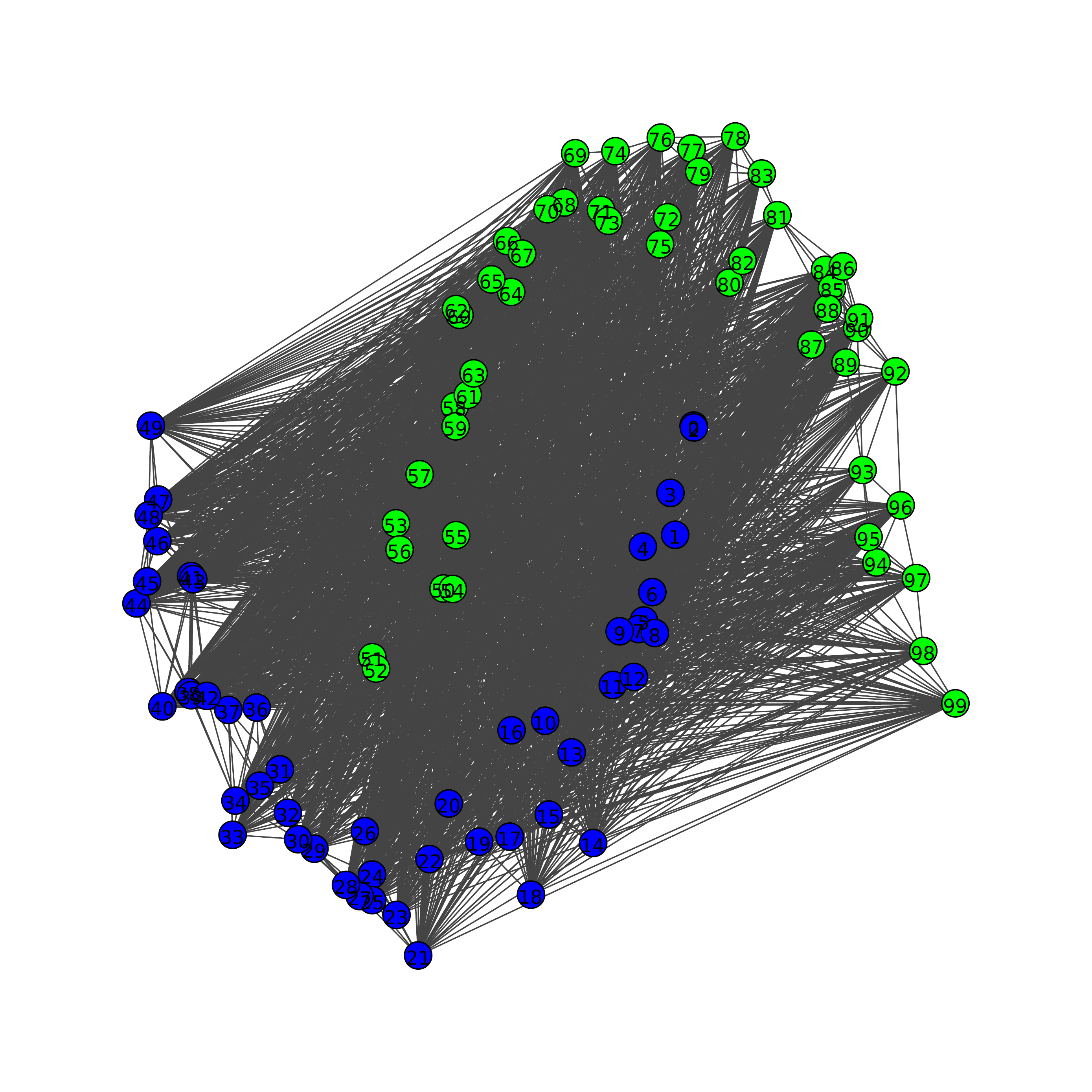}
\caption{Homogeneous graph}
\label{subfig:moons_homo}
\end{subfigure} \hspace*{-0.15in}
\begin{subfigure}{0.21\textwidth} 
\includegraphics[width=0.9\columnwidth]{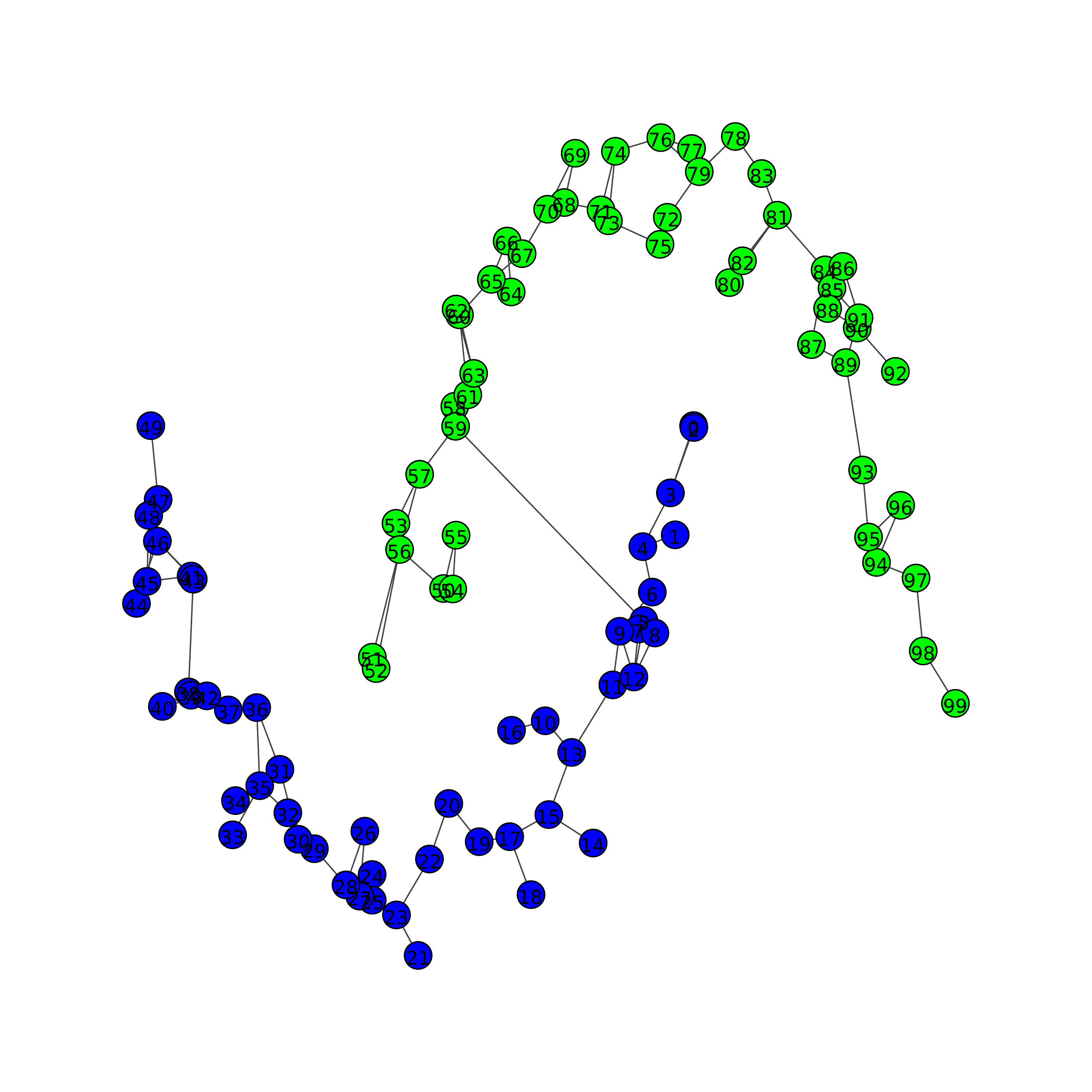}
\caption{\textsc{DBMSTClu}}
\label{subfig:moons_DBMSTClu}
\end{subfigure} \hspace*{-0.15in}
\begin{subfigure}{0.21\textwidth} 
\includegraphics[width=0.9\columnwidth]{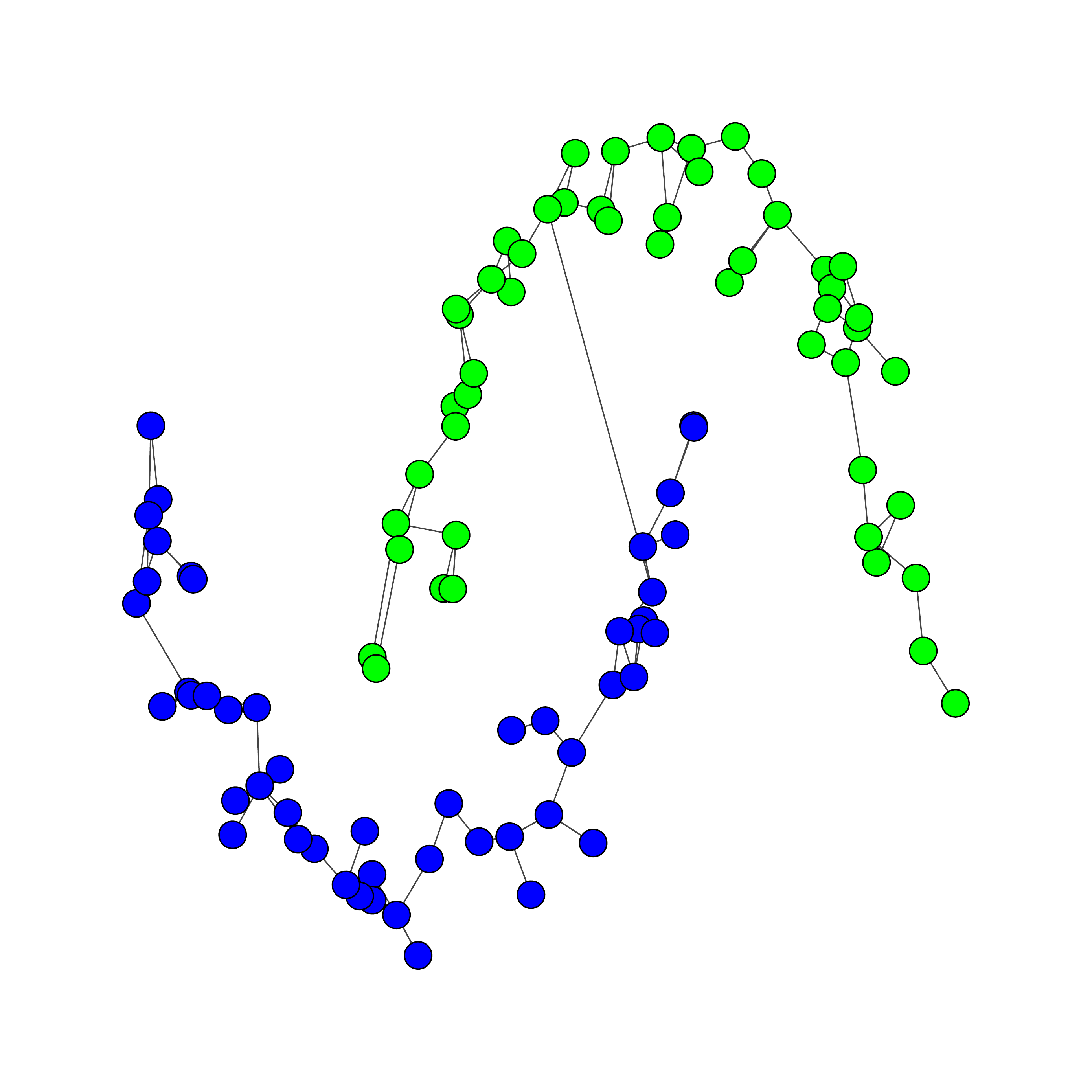}
\caption{\textsc{PTClust}, $\epsilon = 1.0$}
\label{subfig:moons_PTClu1}
\end{subfigure} \hspace*{-0.15in}
\begin{subfigure}{0.21\textwidth} 
\includegraphics[width=0.9\columnwidth]{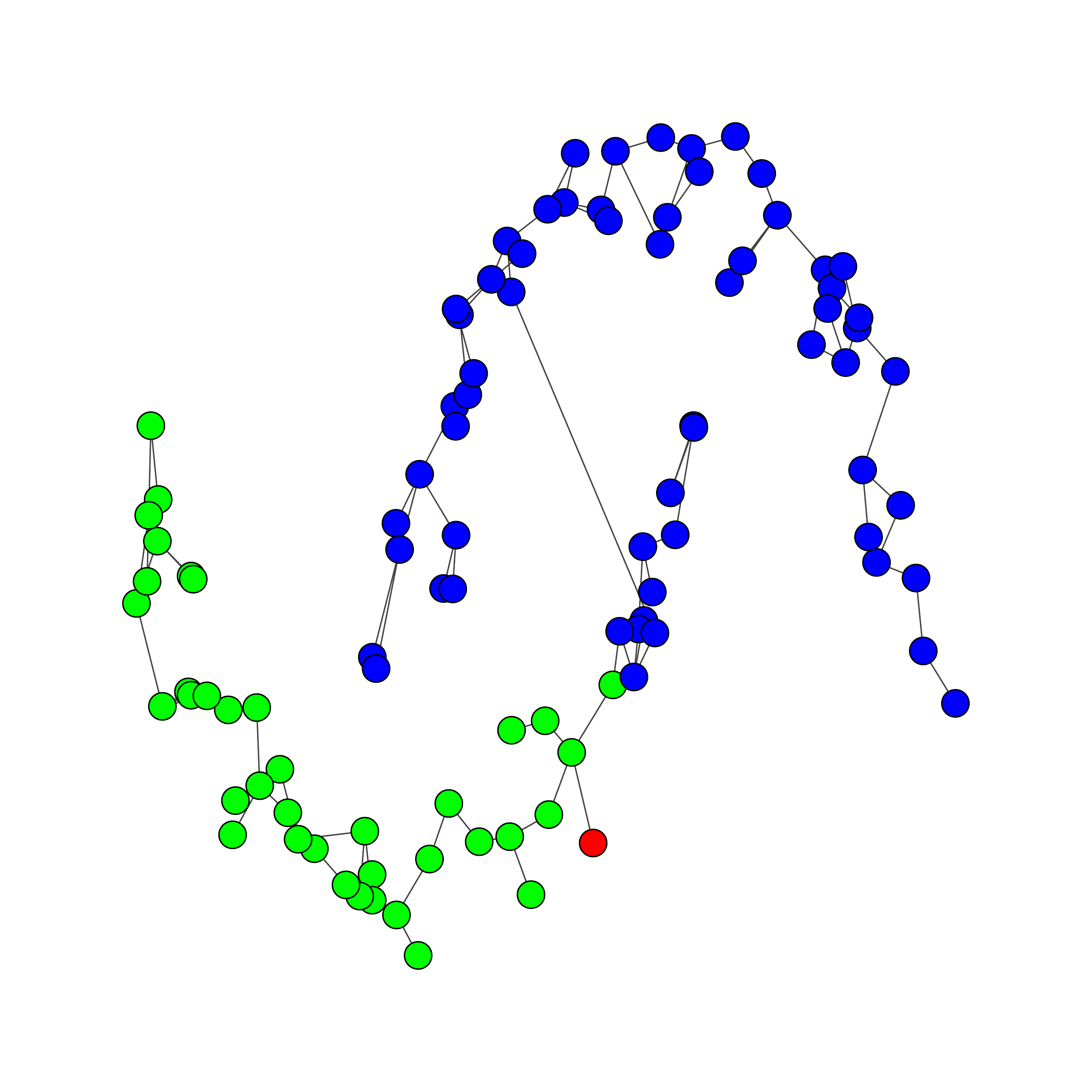}
\caption{\textsc{PTClust}, $\epsilon = 0.7$}
\label{subfig:moons_PTClu07}
\end{subfigure} \hspace*{-0.15in}
\begin{subfigure}{0.21\textwidth} 
\includegraphics[width=0.9\columnwidth]{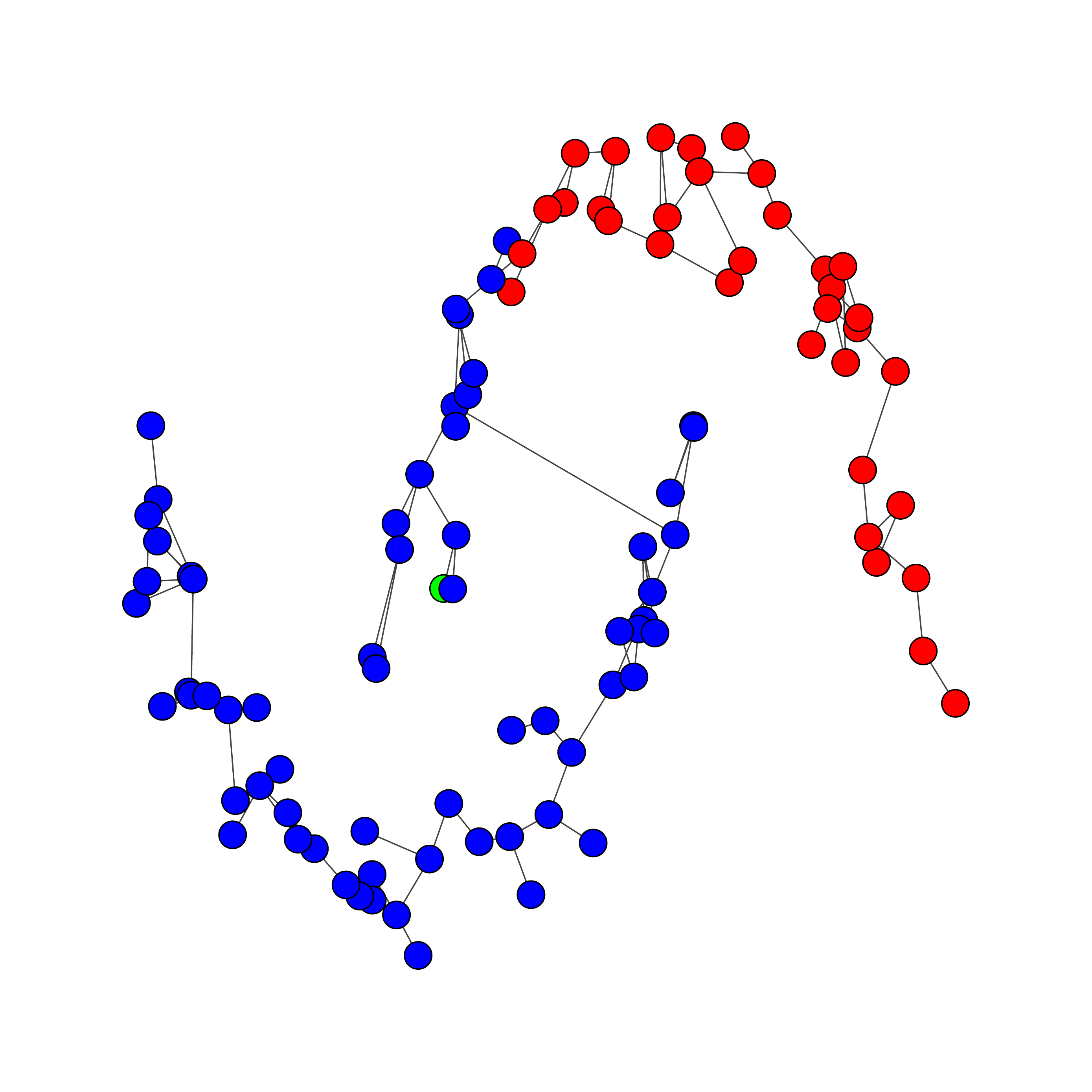}
\caption{\textsc{PTClust}, $\epsilon = 0.5$}
\label{subfig:moons_PTClu05}
\end{subfigure}
\caption{Moons experiments for $n = 100$. \textsc{PTClust} parameters: $w_{min} = 0.1$, $w_{max} = 0.3$, $\mu = 0.1$.}
\label{fig:moons}
\end{figure*}

 So far we have exhibited the trade-off between clustering accuracy and privacy and we experimentally illustrate it with some qualitative results. Let us discuss hereafter the quantitative performances of our algorithm. 
We have performed experiments on two classical synthetic graph datasets for clustering with nonconvex shapes: two concentric circles and two moons, both in their noisy versions. For the sake of readability and for visualization purposes, both graph datasets are embedded into a two dimensional Euclidean space. 
Each dataset contains $100$ data nodes that are represented by a point of two coordinates. Both graphs have been built with respect to the strong homogeneity condition: edge weights within clusters are between $w_{min} = 0.1$ and $w_{max} = 0.3$ while edges between clusters have a weight strictly above $w_{max}^2 / w_{min} = 0.9$.  In practice, the complete graph has trimmed from its irrelevant edges ({\em i.e.} not respecting the strong homogeneity condition). Hence, those graphs are not necessarily Euclidean since close nodes in the visual representation may not be connected in the graph. Finally, weights are normalized between $0$ and $1$.

Figures~\ref{fig:circles} and~\ref{fig:moons} (best viewed in color) show for each dataset (a) the original homogeneous graph $\mathcal{G}$ built by respecting the homogeneity condition, (b) the clustering partition\footnote{For the sake of clarity, the edges in those Figures are represented based on the original weights and not on the privately released weights.} of \textsc{DBMSTClu} with the used underlying MST, 
the clustering partitions for \textsc{PTClust} with $\mu = 0.1$ obtained respectively  with different privacy degrees\footnote{Note that, although the range of $\epsilon$ is in $\mathbb{R}^\star_+$, it is usually chosen in practice in $(0,1]$~\cite[Chap~1\&2]{Dwork_2013}.} : $\epsilon = 0.5$ (c), $\epsilon = 0.7$ (c) and $\epsilon = 1.0$ (e). The utility function $u_{\mathcal{G}}$ corresponds to the graph weight.
Each experiment is carried out independently and the tree topology obtained by \textsc{PAMST} will eventually be different. This explains why the edge between clusters may not be the same when the experiment is repeated with a different level of privacy. However, this will marginally affect the overall quality of the clustering.

As expected, \textsc{DBMSTClu} recovers automatically the right partition and the results are shown here for comparison with \textsc{PTClust}. For \textsc{PTClust}, the true MST is replaced with a private approximate MST obtained for suitable $\tau$ and $p$ ensuring final weights between $0$ and $1$.

When the privacy degree is moderate ($\epsilon \in \{ 1.0, 0.7\}$), it appears that the clustering result is slightly affected. More precisely, in Figures~\ref{subfig:circles_PTClu1} and \ref{subfig:circles_PTClu07} the two main clusters are recovered while one point is isolated as a singleton. This is due to the randomization involved in determining the edge weights for the topology returned by \textsc{PAMST}.
In Figure~\ref{subfig:moons_PTClu1}, the clustering is identical to the one from \textsc{DBMSTClu} in Figure~\ref{subfig:moons_DBMSTClu}. In Figure~\ref{subfig:circles_PTClu07}, the clustering is very similar to the DBMSTClu one, with the exception of an isolated singleton.
However, as expected from our theoretical results, when $\epsilon$ is decreasing, the clustering quality deteriorates, as \textsc{DBMSTClu} is sensitive to severe changes in the MST (cf. Figure~\ref{subfig:circles_PTClu05}, \ref{subfig:moons_PTClu05}).




\section{Conclusion}

In this paper, we introduced \textsc{PTClust}, a novel graph clustering algorithm able to recover arbitrarily-shaped clusters while preserving differential privacy on the weights of the graph. It is based on the release of a private approximate minimum spanning tree of the graph of the dataset, by performing suitable cuts to reveal the clusters. To the best of our knowledge, this is the first differential private graph-based clustering algorithm adapted to nonconvex clusters.
The theoretical analysis exhibited a trade-off between the degree of privacy and the accuracy of the clustering result. This work suits to applications where privacy is a critical issue and it could pave the way to metagenomics and genes classification using individual gene maps while protecting patient privacy. Future work will be devoted to deeply investigate these applications.

\bibliography{biblio}
\bibliographystyle{abbrvnat}

\newpage

\section*{SUPPLEMENTARY MATERIAL}
\section{Proof regarding the accuracy of \textsc{DBMSTClu}}

\subsection{Proof of Theorem~\ref{thm:thmA}}

This theorem relies on the following lemma:

\begin{lemma} \label{lem:heaviest_edge_inCut}
Let us consider a graph $\mathcal{G} = (V, E, w)$ with $K$ clusters $C^*_1, \ldots, C^*_K$ and $\mathcal{T}$ an MST of $\mathcal{G}$. 
If for all $i \in [K]$, $C_i^*$ is  weakly homogeneous, then $\underset{e \in \mathcal{T}}{\argmax} \ w(e) \subset Cut_{\mathcal{G}}(\mathcal{T})$ i.e. the heaviest edges in $\mathcal{T}$ are in $Cut_{\mathcal{G}}(\mathcal{T})$.
\end{lemma}
\begin{proof}
Let us consider $C^*_i$ a cluster of $\mathcal{G}$. As $C^*_i$ is weakly homogeneous, 
$\forall j \in [K]$ s.t. $e^{(ij)} \in Cut_{\mathcal{G}}(\mathcal{T})$, $\underset{e \in \mathcal{T}_{|C^*_i}}{\max} \ w(e) < w(e^{(ij)})$.
Hence, $\underset{e \in E(\mathcal{T})}{\argmax} \ w(e) \subset Cut_{\mathcal{G}}(\mathcal{T})$.
\end{proof}

\begin{theorem*}~\textbf{\ref{thm:thmA}}
Let us consider a graph $\mathcal{G} = (V, E, w)$ with $K$ homogeneous clusters $C^*_1, \ldots, C^*_K$ and $\mathcal{T}$ an MST of $\mathcal{G}$.
Let now assume that at step $k < K-1$, \textsc{DBMSTClu} built $k+1$ subtrees $\mathcal{C}_1, \ldots, \mathcal{C}_{k+1}$ by cutting $e_1, \ e_2, \ \ldots, \ e_k \in E$.

Then, $Cut_k := Cut_{\mathcal{G}}(\mathcal{T}) \ \backslash \ \{ e_1, \ e_2, \ \ldots, \ e_k \} \neq \emptyset \implies \DBCVI_{k+1} \geq DBCVI_k$, i.e. if there are still edges in $Cut_k$, the algorithm will continue to perform some cut.
\end{theorem*}

\begin{proof}
Let note DBCVI at step $k$, $DBCVI_k = \sum_{i = 1}^{k+1} \frac{ |\mathcal{C}_i|}{N} V_C(\mathcal{C}_i)$. Let assume that $Cut_k \neq \emptyset$. Therefore, there is $e^* \in Cut_k$ and $i \in \{1, \ \ldots, \ k+1 \}$ s.t. $e^* \in E(\mathcal{C}_i)$.
Since $e^* \in Cut_{\mathcal{G}}(\mathcal{T})$, using Lem.~\ref{lem:heaviest_edge_inCut}, one can always take $e^* \in \underset{e \in E(\mathcal{C}_i)}{\argmax} \ w(e)$. Then, if we denote $\mathcal{C}_i^1$, $\mathcal{C}_i^2$ the two subtrees of $\mathcal{C}_i$ induced by the cut of $e^*$ (see Fig.~\ref{fig:proofThA} for an illustration) and $DBCVI_{k+1}(e^*)$ the associated DBCVI value,  
\begin{align}
\Delta &= DBCVI_{k+1}(e^*) - DBCVI_k \notag \\
&= \frac{|\mathcal{C}_i^1|}{N} \underbrace{\left( \frac{\SEP(\mathcal{C}_i^1) - \DISP(\mathcal{C}_i^1)}{ \max(\SEP(\mathcal{C}_i^1),\DISP(\mathcal{C}_i^1))} \right)}_{V_C(\mathcal{C}_i^1)} + \frac{|\mathcal{C}_i^2|}{N} \underbrace{\left( \frac{\SEP(\mathcal{C}_i^2) - \DISP(\mathcal{C}_i^2)}{ \max(\SEP(\mathcal{C}_i^2),\DISP(\mathcal{C}_i^2))} \right)}_{V_C(\mathcal{C}_i^2)} - \frac{|\mathcal{C}_i|}{N} \underbrace{\left( \frac{\SEP(\mathcal{C}_i) - \DISP(\mathcal{C}_i)}{ \max(\SEP(\mathcal{C}_i),\DISP(\mathcal{C}_i))} \right)}_{V_C(\mathcal{C}_i)}. \notag
\end{align}
There are two possible cases:
\begin{enumerate}
\item $V_C(\mathcal{C}_i) \leq 0$, then $\SEP(\mathcal{C}_i) \leq \DISP(\mathcal{C}_i) = w(e^*)$. 
As for $l \in \{1, 2 \}$, $\SEP(\mathcal{C}_i^l) \geq \SEP(\mathcal{C}_i)$ and $\DISP(\mathcal{C}_i^l) \leq \DISP(\mathcal{C_i})$ because $e^* \in \underset{e \in E(\mathcal{C_i})}{\argmax} \ w(e)$, then, for $l \in \{1, 2 \}$, 
$$\frac{ \SEP(\mathcal{C}_i^l) - \DISP(\mathcal{C}_i^l) }{ \max(\SEP(\mathcal{C}_i^l),  \DISP(\mathcal{C}_i^l))} \geq \frac{ \SEP(\mathcal{C}_l) - \DISP(\mathcal{C}_i) }{ \max(\SEP(\mathcal{C}_i),  \DISP(\mathcal{C}_i))} = \frac{\SEP(\mathcal{C}_i)}{ w(e)} - 1$$ and $\Delta \geq 0$.
\item $V_C(\mathcal{C}_i) \geq 0$, then $\SEP(\mathcal{C}_i) \geq \DISP(\mathcal{C}_i) = w(e^*)$ i.e. $\max( \SEP(\mathcal{C}_i), \DISP(\mathcal{C}_i)) = \SEP(\mathcal{C}_i)$, for $l \in \{1, 2 \}$, $\DISP(\mathcal{C}_i^l) \leq \DISP(\mathcal{C}_i)$ i.e. $\DISP(\mathcal{C}_i^l) \leq w(e^*)$, $\SEP(\mathcal{C}_i^l) = w(e^*)$ hence $\SEP(\mathcal{C}_i^l) \geq \DISP(\mathcal{C}_i^l)$.
Thus, $V_C(\mathcal{C}_i) = 1 - \frac{\DISP(\mathcal{C}_i)}{\SEP(\mathcal{C}_i)}$ and for $l \in \{1, 2 \}$, $V_C(\mathcal{C}_i^l) = 1 - \frac{ \DISP(\mathcal{C}_i^l)}{\SEP(\mathcal{C}_i^l)}$. Then, for $l \in \{1, 2 \}$, $V_C(\mathcal{C}_i^l) \geq V_C(\mathcal{C}_i)$ and $\Delta \geq 0$.

\end{enumerate}
For both cases, $\Delta = DBCVI_{k+1}(e^*) - DBCVI_k \geq 0$. Hence, at least the cut of $e^*$ improves the current DBCVI, so the algorithm will perform a cut at this stage.

\end{proof}

\begin{figure}[ht]
\vskip 0.2in
\begin{center}
\centerline{\includegraphics[width=0.6\columnwidth]{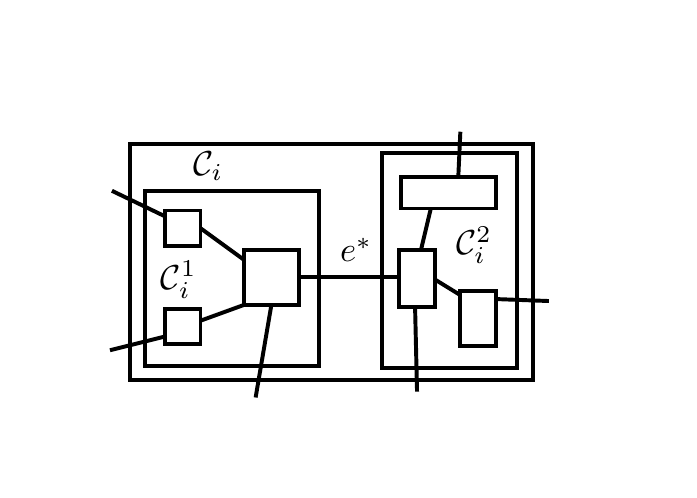}}
\vskip -0.4in
\caption{Illustration for Th.~\ref{thm:thmA}'s proof.}
\label{fig:proofThA}
\end{center}
\vskip -0.2in
\end{figure}

\subsection{Proof of Theorem~\ref{thm:thmB}}

\begin{theorem*}~\textbf{\ref{thm:thmB}}
Let us consider a graph $\mathcal{G} = (V, E, w)$ with $K$ homogeneous clusters $C^*_1, \ldots, C^*_K$ and $\mathcal{T}$ an MST of $\mathcal{G}$.

Let now assume that at step $k < K - 1$, \textsc{DBMSTClu} built $k+1$ subtrees $\mathcal{C}_1, \ldots, \mathcal{C}_{k+1}$ by cutting $e_1, \ e_2, \ \ldots, \ e_k \in E$. We still denote $Cut_k := Cut_{\mathcal{G}}(\mathcal{T}) \backslash \{ e_1, \ e_2, \ \ldots, \ e_k \}$.

Then, $Cut_k \neq \emptyset \implies \underset{e \in \mathcal{T} \backslash \{ e_1, \ e_2, \ \ldots, \ e_k \}}{\argmax} DBCVI_{k+1}(e) \subset Cut_k$ i.e. the edge that the algorithm cuts at step $k+1$ is in $Cut_k$. 
\end{theorem*}

\begin{proof}
It is sufficient to show that, at step $k$, if there exists an edge $e^*$ whose cut builds two clusters, then $e^*$ maximizes DBCVI among all possible cuts in the union of itself and both resulting clusters. 
Indeed, showing this for two clusters, one can easily generalize to the whole graph as a combination
of couples of clusters (see Fig.~\ref{fig:proofThBdouble} for an illustration): if for each couple, the best local solution is in $Cut_k$, then the best general solution is necessary in $Cut_k$.

Let us consider at step $k$ of the algorithm two clusters $C^*_1$ and $C^*_2$ such that $e^*$ the edge separating them in $\mathcal{T}$ is in $Cut_k$ (see Fig.~\ref{fig:proofThB} for an illustration). For readability we denote $\mathcal{T}_{|C^*_1}=\mathcal{C}^*_1$ and $\mathcal{T}_{|C^*_2}=\mathcal{C}^*_2$
Let us proof that for all $\tilde{e} \in \mathcal{T}_{|C^*_1 \cup C^*_2}$, one has:
$DBCVI_{k+1}(e^*) > DBCVI_{k+1}(\tilde{e})$.
W.l.o.g. let assume $\tilde{e} \in \mathcal{C}^*_1$ and let denote $\mathcal{C}^*_{1,1}$ and $\mathcal{C}^*_{1,2}$ the resulting subtrees from the cut of $\tilde{e}$. We still denote $DBCVI_{k+1}(e)$ the value of the DBCVI at step $k+1$ for the cut of $e$.
\begin{align*}
\Delta &:= DBCVI_{k+1}(e^*) - DBCVI_{k+1}(\tilde{e}) \\ 
&= \underbrace{\frac{|\mathcal{C}^*_1|}{N} \left( \frac{\SEP(\mathcal{C}^*_1) - \DISP(\mathcal{C}^*_1)}{ \max( \SEP(\mathcal{C}^*_1), \DISP(\mathcal{C}^*_1))} \right) + \frac{|\mathcal{C}^*_2|}{N} \left( \frac{\SEP(\mathcal{C}^*_2) - \DISP(\mathcal{C}^*_2)}{ \max( \SEP(\mathcal{C}^*_2), \DISP(\mathcal{C}^*_2))} \right)}_{A} \\
&- \underbrace{\left( \frac{|\mathcal{C}^*_{1,1}|}{N} \left( \frac{\SEP(\mathcal{C}^*_{1,1}) - \DISP(\mathcal{C}^*_{1,1})}{ \max( \SEP(\mathcal{C}^*_{1,1}), \DISP(\mathcal{C}^*_{1,1}))} \right) + \frac{|\mathcal{C}^*_{1,2}|}{N} \left( \frac{\SEP(\mathcal{C}^*_{1,2}) - \DISP(\mathcal{C}^*_{1,2})}{ \max( \SEP(\mathcal{C}^*_{1,2}), \DISP(\mathcal{C}^*_{1,2}))} \right) \right)}_{B}
\end{align*}
By weak homogeneity of $C_1^*$ and $C_2^*$, $A = \frac{|\mathcal{C}^*_1|}{N} \left( 1 - \frac{\DISP(\mathcal{C}^*_1)}{ \SEP(\mathcal{C}^*_1)} \right) + \frac{|\mathcal{C}^*_2|}{N} \left( 1 - \frac{\DISP(\mathcal{C}^*_2)}{ \SEP(\mathcal{C}^*_2)} \right) > 0$ 

$B = \underbrace{ \frac{|\mathcal{C}^*_{1,1}|}{N} \left( \frac{\SEP(\mathcal{C}^*_{1,1}) - \DISP(\mathcal{C}^*_{1,1})}{ \max( \SEP(\mathcal{C}^*_{1,1}), \DISP(\mathcal{C}^*_{1,1}))} \right)}_{B_1} +  \underbrace{\frac{|\mathcal{C}^*_{1,2}|}{N} \left( \frac{\SEP(\mathcal{C}^*_{1,2}) - \DISP(\mathcal{C}^*_{1,2})}{ \max( \SEP(\mathcal{C}^*_{1,2}), \DISP(\mathcal{C}^*_{1,2}))} \right)}_{B_2}$

By Lem.~\ref{lem:heaviest_edge_inCut}, $e^* \in \underset{e \in E(\mathcal{T}_{|C^*_1 \cup C^*_2})}{\argmax} w(e)$ so $\DISP(\mathcal{C}^*_{1,2}) = w(e^*)$. 

Since $e^* \in Cut_{\mathcal{G}}(\mathcal{T})$, one has $w(e^*) \geq \max( \SEP(\mathcal{C}^*_1), \SEP(\mathcal{C}^*_2))$. Moreover, as $\mathcal{C}^*_2$ is a subtree of $\mathcal{C}^*_{1,2}$, then $\SEP(\mathcal{C}^*_{1,2}) \leq \SEP(\mathcal{C}^*_2)$. Thus, $w(e^*) \geq \SEP(\mathcal{C}^*_{1,2})$. 
Finally, $B_2 = \frac{|\mathcal{C}^*_{1,2}|}{N} \left( \frac{\SEP(\mathcal{C}^*_{1,2})}{ \DISP(\mathcal{C}^*_{1,2})}-1 \right) \leq 0$.

Besides, $w(\tilde{e}) \leq \SEP(\mathcal{C}_1^*) \implies \SEP(\mathcal{C}^*_{1,1}) = w(\tilde{e}) \leq \underset{e \in E(\mathcal{C}^*_1)}{\max} \ w(e)$ and $\DISP(\mathcal{C}^*_{1,1}) = \underset{e \in E(\mathcal{C}^*_{1,1})}{\max} \ w(e) \geq \underset{e \in E(\mathcal{C}^*_1)}{\min} w(e)$. 
Then, two possibilities hold:
\begin{enumerate}
\item $B_1 < 0 \implies B < 0 < A$.  
\item $B_1 \geq 0$, thus one has $B_1 = \frac{|\mathcal{C}^*_{1,1}|}{N} \left( 1 - \frac{\DISP(\mathcal{C}^*_{1,1})}{\SEP(\mathcal{C}^*_{1,1})} \right) \leq \frac{|\mathcal{C}^*_{1,1}|}{N} \left( 1 - \frac{\underset{e \in \mathcal{C}^*_1}{\min} w(e)}{\underset{e \in \mathcal{C}^*_1}{\max} w(e)} \right)$. 
Under weak homogeneity condition, there is:
$\frac{\DISP(C^*_1)}{\SEP(C^*_1)} < \frac{\underset{e \in \mathcal{C}^*_1}{\min} w(e)}{\underset{e \in \mathcal{C}^*_1}{\max} w(e)}$. Thus, 
\begin{align*}
B_1 &< \frac{|C^*_{1,1}|}{N} \left( 1 - \frac{ \DISP(\mathcal{C}^*_1)}{ \SEP(\mathcal{C}^*_1)}\right) \\
&< \frac{|\mathcal{C}^*_{1}|}{N} \left( 1 - \frac{ \DISP(\mathcal{C}^*_1)}{ \SEP(\mathcal{C}^*_1)}\right) \mbox{ because } \mathcal{C}^*_{1,1} \mbox{ is a subtree of} \mathcal{C}^*_1 \\
&< A 
\end{align*}
So, $B_1 + B_2 = B < A = DBCVI_{k+1}(e^*)$. 
\end{enumerate}
Since $B < A$, $\Delta > 0$ and $e^*$ maximizes DBCVI among all possible cuts in the union of itself and both resulting clusters. Q.E.D.
\end{proof}

\begin{figure}[ht]
\vskip 0.2in
\begin{center}
\centerline{\includegraphics[width=0.6\columnwidth]{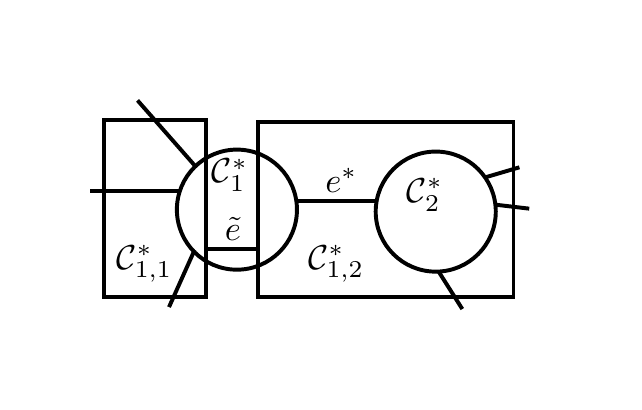}}
\vskip -0.4in
\caption{Illustration for Th.~\ref{thm:thmB}'s proof.}
\label{fig:proofThB}
\end{center}
\vskip -0.2in
\end{figure}

\begin{figure}[ht]
\vskip 0.2in
\begin{center}
\centerline{\includegraphics[width=0.6\columnwidth]{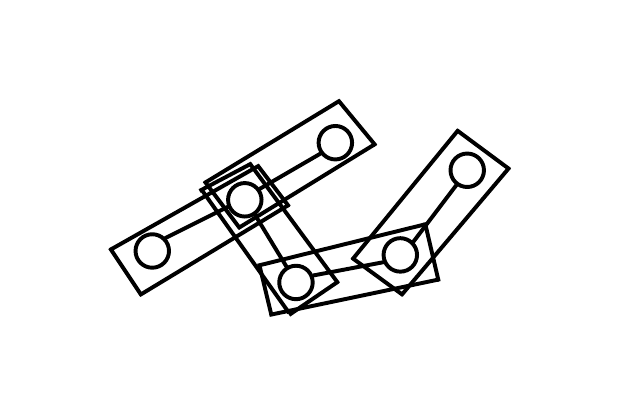}}
\vskip -0.4in
\caption{Illustration for Th.~\ref{thm:thmB}'s proof. Each circle corresponds to a cluster. The six clusters are handled within five couples of clusters.}
\label{fig:proofThBdouble}
\end{center}
\vskip -0.2in
\end{figure}

\subsection{Proof of Theorem~\ref{thm:thmC}}

\begin{theorem*}~\textbf{\ref{thm:thmC}}
Let us consider a graph $\mathcal{G} = (V, E, w)$ with $K$ weakly homogeneous clusters $C^*_1, \ldots, C^*_K$ and $\mathcal{T}$ an MST of $\mathcal{G}$.
Let now assume that at step $K-1$, \textsc{DBMSTClu} built $K$ subtrees $\mathcal{C}_1, \ldots, \mathcal{C}_{K}$ by cutting $e_1, \ e_2, \ \ldots, \ e_{K-1} \in E$. We still denote $Cut_{K-1} := Cut_{\mathcal{G}}(\mathcal{T}) \backslash \{ e_1, \ e_2, \ \ldots, \ e_{K-1} \}$.
 
Then, for all $e \in \mathcal{T} \backslash \{ e_1, \ e_2, \ \ldots, \ e_{K-1} \}$, $DBCVI_{K}(e) < DBCVI_{K-1}$ i.e. the algorithm stops: no edge gets cut during step $K$.

\begin{proof}
According to Th.~\ref{thm:thmA} and Th.~\ref{thm:thmB}, for all $k < K$, if $Cut_{k} \neq \emptyset$, the algorithm performs some cut from $Cut_{\mathcal{G}}(\mathcal{T})$. We still denote for all $j\in [K]$ $\mathcal{C}^*_j=\mathcal{T}_{|C^*_j}$.Since $|Cut_{\mathcal{G}}(\mathcal{T})| = K-1$, the $K-1$ first steps produce $K-1$ cuts from $Cut_{\mathcal{G}}(\mathcal{T})$. Therefore, $\DBCVI_{K-1} = \underset{j \in [K-1]}{\sum} \frac{|\mathcal{C}^*_j|}{N} V_C(\mathcal{C}^*_j)$.

Let be $e$ the (expected) edge cut at step $K$, splitting the tree $\mathcal{C}^*_i$ into $\mathcal{C}^*_{i,1}$ and $\mathcal{C}^*_{i,2}$. 
\begin{align*}
& \Delta = \DBCVI_{K-1} - \DBCVI_K \\
&= \frac{|\mathcal{C}^*_i|}{N} V_C(\mathcal{C}^*_i) - \frac{|\mathcal{C}^*_{i,1}|}{N} V_C(\mathcal{C}^*_{i,1}) - \frac{|\mathcal{C}^*_{i,2}|}{N} V_C(\mathcal{C}^*_{i,2}) \\
&= \frac{|\mathcal{C}^*_i|}{N} \frac{ \SEP(\mathcal{C}^*_i) - \DISP(\mathcal{C}^*_i)}{\max( \SEP(\mathcal{C}^*_i), \DISP(\mathcal{C}^*_i))} - \frac{|\mathcal{C}^*_{i,1}|}{N} \frac{ \SEP(\mathcal{C}^*_{i,1}) - \DISP(\mathcal{C}^*_{i,1})}{\max( \SEP(\mathcal{C}^*_{i,1}), \DISP(\mathcal{C}^*_{i,1}))} - \frac{|\mathcal{C}^*_{i,2}|}{N} \frac{ \SEP(\mathcal{C}^*_{i,2}) - \DISP(\mathcal{C}^*_{i,2})}{\max( \SEP(\mathcal{C}^*_{i,2}), \DISP(\mathcal{C}^*_{i,2}))}
\end{align*}
Since $C^*_i$ is a weakly homogeneous cluster, therefore $\SEP(\mathcal{C}^*_i) \geq \DISP(\mathcal{C}^*_i)$. 
Then, minimal value of $\Delta$, $\Delta_{min}$ is reached when $\SEP(\mathcal{C}^*_{i,1}) \geq \DISP(\mathcal{C}^*_{i,1})$, $\SEP(\mathcal{C}^*_{i,2}) \geq \DISP(\mathcal{C}^*_{i,2})$, $\SEP(\mathcal{C}^*_{i,1}) = \SEP(\mathcal{C}^*_{i,2}) = \underset{e' \in E(\mathcal{C}^*_i)}{\min} w(e')$, $\DISP(\mathcal{C}^*_{i,1}) = \DISP(\mathcal{C}^*_{i,2}) = \underset{e' \in E(\mathcal{C}^*_i)}{\max} w(e')$. Then,
\begin{align*}
N \times \Delta_{min} &= |\mathcal{C}^*_i| \left(1 - \frac{\DISP(\mathcal{C}^*_i)}{\SEP(\mathcal{C}^*_i)} \right) - |\mathcal{C}^*_{i,1}| \left(1 - \frac{\DISP(\mathcal{C}^*_{i,1})}{\SEP(\mathcal{C}^*_{i,1})} \right) - |\mathcal{C}^*_{i,2}| \left( 1 - \frac{\DISP(\mathcal{C}^*_{i,2})}{\SEP(\mathcal{C}^*_{i,2})} \right) \\
&= |\mathcal{C}^*_i| \left(1 - \frac{\DISP(\mathcal{C}^*_i)}{\SEP(\mathcal{C}^*_i)} \right) - |\mathcal{C}^*_{i,1}| \left(1 - \frac{\underset{e' \in E(\mathcal{C}^*_i)}{\max} w(e')}{\underset{e' \in E(\mathcal{C}^*_i)}{\min} w(e')} \right) - |\mathcal{C}^*_{i,2}| \left( 1 - \frac{\underset{e' \in E(\mathcal{C}^*_i)}{\max} w(e')}{\underset{e' \in E(\mathcal{C}^*_i)}{\min} w(e')} \right) \\
&= |\mathcal{C}^*_i| \left( - \frac{\DISP(\mathcal{C}^*_i)}{\SEP(\mathcal{C}^*_i)}   + \frac{\underset{e' \in E(\mathcal{C}^*_i)}{\max} w(e')}{\underset{e' \in E(\mathcal{C}^*_i)}{\min} w(e')} \right)
\end{align*}
By weak homogeneity condition on $C^*_i$,
$\frac{\DISP(\mathcal{C}^*_i)}{\SEP(\mathcal{C}^*_i)} < \frac{\underset{e' \in E(\mathcal{C}^*_i)}{\min} w(e'))}{\underset{e' \in E(\mathcal{C}^*_i)}{\max} w(e')} \leq \frac{\underset{e' \in E(\mathcal{C}^*_i)}{\max} w(e'))}{\underset{e' \in E(\mathcal{C}^*_i}){\min} w(e')}$. Therefore, $\Delta_{min} > 0$ and $\Delta > 0$. 
\end{proof}
\end{theorem*}

\section{Proofs regarding the accuracy of \textsc{PTClust}}
\subsection{Proof of Theorem~\ref{thm:thmD}}

\begin{theorem*}~\textbf{\ref{thm:thmD}} 
Let us consider a graph $\mathcal{G} = (V, E, w)$ with $K$ strongly homogeneous clusters $C^{*}_1, \ldots, C^{*}_K$ and $\mathcal{T} = \textsc{PAMST}( \mathcal{G}, u_{\mathcal{G}}, w, \epsilon)$, $\epsilon > 0$.
$\mathcal{T}$ has a partitioning topology with probability at least
\begin{equation*}
1 - \sum^{K}_{i = 1} ( |C^*_i| - 1) e^{ - \frac{ \epsilon}{2 \Delta u_{\mathcal{G}} (|V| - 1)} ( \bar{\alpha}_i \underset{e \in  E(\mathcal{G}_{|C^*_i}) }{\max( w(e))} - \underset{e \in  E(\mathcal{G}_{|C^*_i}) }{\min{(w(e))}} ) + \ln( |E|) }
\end{equation*}
\end{theorem*}

\begin{proof}
Let $\mathcal{T} = \textsc{PAMST}( \mathcal{G}, u_{\mathcal{G}}, w, \epsilon)$, $\{\mathcal{R}_{1},..., \mathcal{R}_{|V|-1}\}$ denotes the ranges used in the successive calls of the Exponential mechanism in \textsc{PAMST}$(\mathcal{G},u_{\mathcal{G}},w,\epsilon)$, $r_k=\mathcal{M}_{Exp}(\mathcal{G}, w, u_{\mathcal{G}}, \mathcal{R}_k, \underbrace{\frac{\epsilon}{|V| - 1}}_{\epsilon'} )$,  and $\textnormal{Steps}(C^*_i)$ the set of steps $k$ of the algorithm were $\mathcal{R}_k$ contains at least one edges from $\mathcal{G}_{|C^*_i}$. Finally for readability we denote $u_k= u_{\mathcal{G}}(w,r_k)$
\begin{align*} 
& \mathbb{P}[\mathcal{T} \mbox{ has a partitioning topology}] \\
=& \mathbb{P}[ \underbrace{\forall i,j \in [K], \ i \neq j, \ | \{  (u, v) \in E(\mathcal{T}), \ u \in C^*_i, \ v \in C^*_j \} | = 1}_{A} ] = 1 - \mathbb{P}[\neg A ]
\end{align*}
If we denote $B = `` \forall i \in [K], \ \forall k >1 \in \textnormal{Steps}(C^*_i), \mbox{ if } r_{k-1} \in E(\mathcal{G}_{|C^*_i}) \mbox{ then } r_{k} \in E(\mathcal{G}_{|C^*_i})"$
One easily has: $ B \implies A $, therefore $\mathbb{P}[ \neg A] \leq \mathbb{P}(\neg B)$.
Moreover, by using the privacy/accuracy trade-off of the exponential mechanism, one has $$\forall t  \in \mathbb{R}, \forall i \in [K], \forall k \in \mbox{Steps}(C_i^*) \ \mathbb{P} \left[\underbrace{u_k \leq -\frac{2\Delta u_{\mathcal{G}}}{\epsilon'}(t + \ln|\mathcal{R}_k|)}_{A_{k}(t)} \right] \leq \exp(-t). $$
Moreover one can major $ \mathbb{P}[\neg B] $ as follows
\begin{align*} 
 &\mathbb{P}\left[\exists i \in [K], \exists k \in \textnormal{Steps}(C^*_i) \mbox{ s.t } r_{k-1} \in E(\mathcal{G}_{|C^*_i}) \mbox{ and } r_{k} \notin E(\mathcal{G}_{|C^*_i})\right] \\
\intertext{By using the union bound, one gets }
\leq & \underset{i \in [K]}{\sum} \mathbb{P}\left[ \exists k \in \textnormal{Steps}(C^*_i) \textnormal{ s.t } r_{k-1} \in E(\mathcal{G}_{|C^*_i}) \mbox{ and } r_{k} \notin E(\mathcal{G}_{|C^*_i}) \right] \\
\intertext{Using the strong homogeneity of the clusters, one has }
= &  \underset{i \in [K]}{\sum} \mathbb{P}\left[ \exists k \in \textnormal{Steps}(C^*_i) \textnormal{ s.t } u_k \leq -| \bar{\alpha}_i \underset{e \in  E(\mathcal{G}_{|C^*_i})  }{\max w(e)} -  \underset{ r\in \mathcal{R}_k}{\min w(r)}   | \right] \\
\leq &  \underset{i \in [K]}{\sum} \mathbb{P}\left[ \exists k \in \textnormal{Steps}(C^*_i) \textnormal{ s.t } u_k \leq -| \bar{\alpha}_i \underset{e \in  E(\mathcal{G}_{|C^*_i})  }{\max w(e)} -  \underset{e \in  E(\mathcal{G}_{|C^*_i})  }{\min w(e)}   | \right] \\ 
\intertext{ By setting $t_{k,i} = \frac{ \epsilon'}{2 \Delta u_{\mathcal{G}} } ( \bar{\alpha}_i \underset{e \in  E(\mathcal{G}_{|C^*_i}) }{\max( w(e))} - \underset{e \in  E(\mathcal{G}_{|C^*_i}) }{\min{(w(e))}} ) + \ln( |\mathcal{R}_k|)$ one gets } 
= &  \underset{i \in [K]}{\sum} \mathbb{P}\left[ \exists k \in \textnormal{Steps}(C^*_i) \textnormal{ s.t } A_k(t_{k,i}) \right]\\ 
\intertext{ Since for all $i \in [K]$, and $k \in \textnormal{Steps}(C^*_i)$, $|\mathcal{R}_k| \leq |E|$, and using a union bound, one gets }
\leq & \underset{i \in [K]}{\sum} \ \underset{k \in \textnormal{Steps}(C^*_i)}{\sum \mathbb{P} } \left[ A_k(t_{k,i}) \right]
\leq \underset{i \in [K]}{\sum} \ \underset{k \in \textnormal{Steps}(C^*_i)}{\sum \exp}(-t_{i,k})\\
\leq &  \sum^{K}_{i = 1} ( |C^*_i| - 1) \exp \left( - \frac{ \epsilon}{2 \Delta u_{\mathcal{G}} (|V| - 1)} \left( \bar{\alpha}_i \underset{e \in  E(\mathcal{G}_{|C^*_i}) }{\max w(e)} - \underset{e \in  E(\mathcal{G}_{|C^*_i} ) }{\min{w(e)}} \right) + \ln( |E|) \right) 
\end{align*}      
\end{proof} 

\subsection{Proof of Theorem~\ref{thm:tradeoff}}

Let recall the theorem from S. Kotz on the Laplace distribution and generalizations~(2001):
\begin{theorem} \label{thm:momentorder}
Let $ n \in \mathbb{N}$, $ (X_{i})_{i \in [n] } \underset{iid}{\sim} Lap(\theta,s)$, denoting $X_{r:n}$ the order statistic of rank $r$ one has for all $ k \in \mathbb{N}, $

$$\mathbb{E}\left[\left( X_{r:n} - \theta\right)^k\right] = s^k \frac{n! \Gamma(k+1)}{(r-1)!(n-r)!} \underset{\alpha(n,k)}{\underbrace{\left( (-1)^k \sum_{j=0}^{n-r} a_{j,r,k } + \sum_{j=0}^{r-1} b_{j,r,k} \right)}}$$
\end{theorem}
 
\begin{theorem*}~\textbf{\ref{thm:tradeoff}}
Let us consider a graph $\mathcal{G} = (V, E, w)$ with $K$ strongly homogeneous clusters $C^{*}_1, \ldots, C^{*}_K$ and $T = \textsc{PAMST}( \mathcal{G}, u_{\mathcal{G}}, w, \epsilon)$, and $\mathcal{T'}=\mathcal{M}_{w.r}(T,w_{|T},s,\tau,p)$ with $s<<p,\tau$. Given some cluster $C^{*}_i$, and $j \neq i$ s.t $e^{(ij)} \in Cut_{\mathcal{G}}(\mathcal{T})$, if $H_{\mathcal{T}_{|C_i^*}}(e^{(ij)})$ is verified, then $H_{\mathcal{T'}_{|C_i^*}}(e^{(ij)})$ is verified with probability at least $$ 1 - \frac{\Lambda_1 + (\theta_{(ij)}^2 + \delta) \Lambda_2 - (\Lambda_3^2 +\theta_{(ij)}^2 \Lambda_4^2 )}{\Lambda_1 + (\theta_{(ij)}^2 + \delta) \Lambda_2 + 2\Lambda_3 \Lambda_4} $$ 

with the following notations:
\begin{itemize}
\item $\delta = \frac{s}{p}$, $\theta_{\min}= \frac{\underset{e\in E(\mathcal{T})}{\min} w(e) + \tau}{p}$ 
\item $\theta_{\max}= \frac{\underset{e\in E(\mathcal{T})}{\max} w(e) + \tau}{p}$,$\theta_{(ij)}= \frac{w(e^{(ij)} + \tau}{p}$
\item $ \Lambda_1 = 24 \delta^4 n \alpha(n,4) + 12\theta_{\max}\delta^3n\alpha(n,3)+12\theta_{\max}^2 \delta^2n\alpha(n,2) + 4 \theta_{\max}^3\delta n\alpha(n,1) +\theta_{\max}^4$
\item $ \Lambda_2 =2\delta^2n\alpha(1,2) + 2\theta_{\min}\delta n\alpha(1,1) +\theta_{\min}^2$
\item $ \Lambda_3 =2\delta^2n\alpha(n,2) + 2\theta_{\max}\delta n\alpha(n,1) +\theta_{\max}^2$
\item $ \Lambda_4 = \delta n\alpha(1,1) +\theta_{\min}$
\end{itemize}
\end{theorem*}

\begin{proof}
Let $ \tau >0$ and $p>1$, according to the weight-release mechanism, all the randomized edge weights $w'(e) \mbox{ with }e \in E(\mathcal{T'}) $ are sampled from independents Laplace distributions $Lap( \frac{w(e) + \tau}{p}, \frac{s}{p}).$  Given some cluster $C^{*}_i$, and $j \neq i$ s.t $e^{(ij)} \in Cut_{\mathcal{G}}(\mathcal{T})$, $H_{\mathcal{T}_{|C_i^*}}(e^{(ij)})$ is verified. Finding the probability that $H_{\mathcal{T'}_{|C_i^*}}(e^{(ij)})$ is verified is equivalent to find the probability $\mathbb{P}\left[ \frac{\underset{e\in E\left(\mathcal{C}_i^*\right)}{(\max X_e)^2}}{\underset{e \in E\left(\mathcal{C}_i^*	\right)}{\min X_e}} < X^{out}  \right]$ with $X_e \underset{indep}{\sim} Lap( \frac{w(e) + \tau}{p}, \frac{s}{p}) \mbox{ and } X^{out}\sim Lap( \frac{w(e^{(ij)}) + \tau}{p}, \frac{s}{p}).$ Denoting with $Y_i \underset{iid}{\sim} Lap\left(\theta_{\max}, \delta\right), Z_i \underset{iid}{\sim} Lap\left( \theta_{\min}, \delta\right) \mbox{ and } X^{out}\sim Lap(\theta_{(ij)}, \delta),$ one can lower bounded this probability  by $\mathbb{P}\left[ \frac{\underset{i \in [|C_i^*|-1] }{(\max Y_i)^2}}{\underset{i \in [|C_i^*|-1] }{\min Z_i}} < X^{out}  \right].$ 
Choosing $\tau$ big enough s.t $\underset{i \in [|C_i^*|-1] }{\min Z_i} < 0$ is negligible, one has 
\begin{align*} &\mathbb{P}\left[ \frac{\underset{i \in [|C_i^*|-1] }{(\max Y_i)^2}}{\underset{i \in [|C_i^*|-1] }{\min Z_i}} < X^{out}  \right] \\
= &\mathbb{P}\left[\underset{\varphi}{\underbrace{\underset{i \in [|C_i^*|-1] }{(\max Y_i)^2} - \underset{i \in [|C_i^*|-1] }{\min Z_i} \times X^{out}}} <0  \right]. 
\end{align*}
Moreover since $\tau$, $p$ $> >$  $s$, one has $\mathbb{E}(\varphi) \leq 0.$ Therefore, 
\begin{align*}
\mathbb{P}\left[ \varphi < 0 \right] =& \mathbb{P}\left[ \varphi - \mathbb{E}(\varphi) < \underset{\geq 0}{\underbrace{- \mathbb{E}(\varphi) }}\right] \\ 
=& 1 - \mathbb{P}\left[ \varphi - \mathbb{E}(\varphi) > - \mathbb{E}(\varphi) \right]
\intertext{Using the one-sided Chebytchev inequality, one gets}
\geq & 1 - \frac{\mathbb{V}(\varphi)}{\mathbb{V}(\varphi) + \mathbb{E}(\varphi)^2} =  1 - \frac{\mathbb{V}(\varphi)}{ \mathbb{E}(\varphi^2)}
\end{align*}
By giving an analytic form to $\mathbb{E}(\varphi) \mbox{ and }  \mathbb{V}(\varphi)$ by using Theorem~\ref{thm:momentorder} one gets the expected result.
\end{proof}

\end{document}